%% file: main.tex
\documentclass[prx,twocolumn,floatfix,superscriptaddress,longbibliography,notitlepage]{revtex4-1}

\usepackage[title]{appendix}

\usepackage{graphicx, color, graphpap}
\usepackage{enumitem}
\usepackage{amssymb}
\usepackage{amsthm}
\usepackage{multirow}
\usepackage[colorlinks=true,citecolor=blue,linkcolor=magenta]{hyperref}
\usepackage[T1]{fontenc}
\usepackage{bbm}
\usepackage{thmtools,thm-restate}
\usepackage{verbatim}
\usepackage{mathtools}
\usepackage{titlesec}
\usepackage{amsmath}
\usepackage[title]{appendix}
\usepackage[linesnumbered,ruled,vlined]{algorithm2e}
\SetKwInput{kwInit}{Init}

\usepackage{natbib}

\long\def\ca#1\cb{} 

\newcommand{\ketbra}[2]{| \hspace{1pt} #1 \rangle \langle #2 \hspace{1pt} |}

\newcommand{\norm}[2][]{#1| \! #1| #2 #1| \! #1|}

\newcommand{\ket}[1]{|#1\rangle}               
\newcommand{\bra}[1]{\langle #1|}              
\newcommand{\dya}[1]{\ket{#1}\!\bra{#1}}
\newcommand{\ipa}[2]{\langle #1,#2\rangle}      
\newcommand{\ip}[2]{\langle #1|#2\rangle}      

\newcommand{\rank}{\text{rank}}

\newcommand{\Tr}{{\rm Tr}}

\newcommand{\Var}{{\rm Var}}

\renewcommand{\geq}{\geqslant}
\renewcommand{\leq}{\leqslant}

\newcommand*{\id}{\openone}

{}
{}

\newtheorem{lemma}{Lemma}
\newtheorem{corollary}{Corollary}

\newtheorem{proposition}{Proposition}

\usepackage{color}
\definecolor{cool_green}{rgb}{0.0, 0.5, 0.0}

\newcommand{\tr}{\text{Tr}}
\newcommand{\mc}{\mathcal}
\newcommand{\op}[2]{|#1\rangle\langle #2|}

\usepackage{amssymb}
\usepackage{dsfont}
\usepackage{bbold}

\usepackage[export]{adjustbox}
\theoremstyle{definition}

\usepackage{float}
\usepackage{mathrsfs}
\usepackage{rotating}
\usepackage[linesnumbered,ruled,vlined]{algorithm2e}
\SetKwInput{kwInit}{Init}

\usepackage{titlesec}

\newcommand{\CE}[1]{\mathcal{C}(\ket{#1})}
\newcommand{\Cmax}[1]{\mathcal{C}^*(#1)}
\newcommand{\mbb}{\mathbb}

\newtheorem*{proposition*}{Proposition}
\newtheorem*{corollary*}{Corollary}

\global\csname @topnum\endcsname=1

\begin{document}
    \title{A Hierarchy of Multipartite Correlations Based on Concentratable Entanglement}
    
    \author{Louis Schatzki}
    \email{louisms2@illinois.edu}
    \affiliation{Department of Electrical and Computer Engineering, Coordinated Science Laboratory, University of Illinois at Urbana-Champaign, Urbana, IL 61801, USA}
    \affiliation{Illinois Quantum Information Science and Technology (IQUIST) Center, University of Illinois Urbana-Champaign, Urbana, IL 61801, USA}
     \affiliation{Information Sciences, Los Alamos National Laboratory, Los Alamos, NM 87545, USA}
    
    \author{Guangkuo Liu}
    \affiliation{JILA, University of Colorado/NIST, Boulder, CO, 80309, USA}
    \affiliation{Department of Physics, University of Colorado, Boulder CO 80309, USA}
    
    \author{M. Cerezo}
    \affiliation{Information Sciences, Los Alamos National Laboratory, Los Alamos, NM 87545, USA}
    \affiliation{Quantum Science Center, Oak Ridge, TN 37931, USA}
    
    \author{Eric Chitambar}
    \email{echitamb@illinois.edu}
    \affiliation{Department of Electrical and Computer Engineering, Coordinated Science Laboratory, University of Illinois at Urbana-Champaign, Urbana, IL 61801, USA}
    \affiliation{Illinois Quantum Information Science and Technology (IQUIST) Center, University of Illinois Urbana-Champaign, Urbana, IL 61801, USA}
    \begin{abstract}
        Multipartite entanglement is one of the hallmarks of quantum mechanics and is central to quantum information processing. In this work we show that Concentratable Entanglement (CE), an operationally motivated entanglement measure, induces a hierarchy upon pure states from which different entanglement structures can be experimentally certified. In particular, we find that nearly all genuine multipartite entangled states can be verified through the CE. Interestingly, GHZ states prove to be far from maximally entangled according to this measure. Instead we find the exact maximal value and corresponding states for up to 18 qubits and show that these correspond to extremal quantum error correcting codes. The latter allows us to unravel a deep connection between CE and coding theory. Finally, our results also offer an alternative proof, on up to 31 qubits, that absolutely maximally entangled states do not exist.
    \end{abstract}
    \maketitle
    

    Entanglement is one of the defining properties of quantum mechanics \cite{Einstein1935}, and it has been shown to serve as a fundamental resource for quantum information processing, from cryptography to computation and quantum sensing~\cite{Ekert1998, Wootters1998, Chaves2010, Jozsa2003, Baek1998, Horodecki2009,Datta2007, Raussendorf2001,degen2017quantum,huerta2022inference}. As such, characterizing the entanglement in a state is a fundamental task to determine its utility for information processing.
    
    Bipartite pure states are said to be entangled if they cannot be written as a tensor product $\ket{\psi_A}\otimes\ket{\psi_B}$~\cite{nielsen_chuang_2010}. On the other hand, entanglement in multipartite systems is more complex, as entanglement can arise between certain subsystems but not necessarily across the entire system.  For example, the four-party biseparable state $\ket{\psi_{AB}}\otimes\ket{\psi_{CD}}$ lacks any entanglement across the partition $AB:CD$. States for which no biseparable partition can be drawn are said to have genuine multipartite entanglement (GME) \cite{Huber2014, Toth2005}, and a significant amount of work has been put forward towards quantifying and characterizing GME \cite{Dur2000, Szalay2015, Walter2016, Horodecki2009, Barnum2001, Miyake2003,Wong2001, Coffman2000, Meyer2002, Eisert2001, Brylinski2002}. Yet, even among the collection of non-GME states there is a great deal of complexity, and a non-GME state $\ket{\psi}$ can be classified according to the number of separable cuts it possesses, as well as the number of entangled systems within each cut. 
    
    With the rapid growth of quantum technologies  \cite{Preskill2018,Gyongyosi2019, Corcoles2019} and experimental hardware capable of generating multipartite entangled states, there is a natural demand for methods to decide whether a given $n$-partite state is GME, and if it is not, for how to determine its product structure.  Ideally, one would like to answer these questions using some entanglement measure that is experimentally accessible through simple protocols.  In this work we show that the concentratable entanglement is one such entanglement measure, as it can be used for identifying entanglement structures as well as other interesting features of multipartite entanglement.  

The concentratable entanglement (CE) is an efficiently computable entanglement measure recently introduced in Ref. \cite{Beckey2021}.  For an $n$-qubit pure state $\ket{\psi}$, the CE takes the form  
    \begin{equation}\label{eq:CE}
        \CE{\psi} = 1-\frac{1}{2^n}\sum_{\alpha\in Q}\Tr[\rho_\alpha^2],
    \end{equation}
     where $\alpha$ is a set of qubit labels,  $\rho_\alpha$ is the reduced state of $\ket{\psi}$ on the set of qubits in $\alpha$, and $Q$ is the power set of $\{1,2\ldots,n\}$.  The CE captures entanglement in the system by averaging the reduced state purities across all partitions. Operationally, CE is the probability that at least one SWAP test should fail when $n$ of them are applied in parallel across two copies of $\ket{\psi}$ (see Fig.~\ref{fig:swap}). Since a failed SWAP test leads to the generation of a Bell pair, the CE also quantifies how well entanglement can be concentrated  using two copies of the state and SWAP tests~\cite{Beckey2021}. The controlled SWAP (cSWAP) used to measure CE is a basic building block in quantum communication protocols \cite{Buhrman2001}, and it has been recognized as an experimentally accessible tool for measuring and witnessing entanglement \cite{Gutoski2015, Foulds2021,Beckey2021, Foulds2021}.

    In this work we show  that the CE can only attain a certain maximal value when restricted to $n$-qubit states having a fixed product state structure.  These maximal values naturally induce a hierarchy on the set of all pure states, separating classes of states with different numbers of separable cuts or different numbers of qubits within each cut (for example, see Table \ref{table:5hier}).  Hence, using the CE measured for a given state, one can certify that entanglement must exist between at least a certain number of parties.

    The rest of this paper is devoted to developing this hierarchy and exploring applications.  Our first significant technical challenge is to compute the maximal value of CE among all $n$-qubit states.  This in itself turns out to be a very intriguing problem.  While we do not have a general solution, we provide a linear programming upper bound that is tight for at least up to 18 qubits (except $n\in \{7,13,15,16\}$ where we know of states coming within a few decimal places of the bound). Interestingly, this linear program closely matches those arising in the study of quantum error correcting cods, and we exploit this connection between entanglement and coding theory to unravel a deep connection between CE and coding theory.  By deriving the Haar statistics of the CE, we then show that most states have near maximal CE and our hierarchy is tight enough to certify GME in nearly all such states. Finally, we provide rigorous connections between CE and cSWAP to other entanglement measures.  All results and discussion are presented below, with detailed proofs delayed to the supplemental material (SM)~\footnote{See Supplemental Material which contains additional details and proofs as well as Refs. \cite{Werner1989,PhysRevA.63.042111,spiegel_lipschutz_liu_2018,Hein2006,Varbanov2007,1266813,BACHOC200155} }.
     \begin{figure}[t]
         \centering
         \includegraphics[scale=.75]{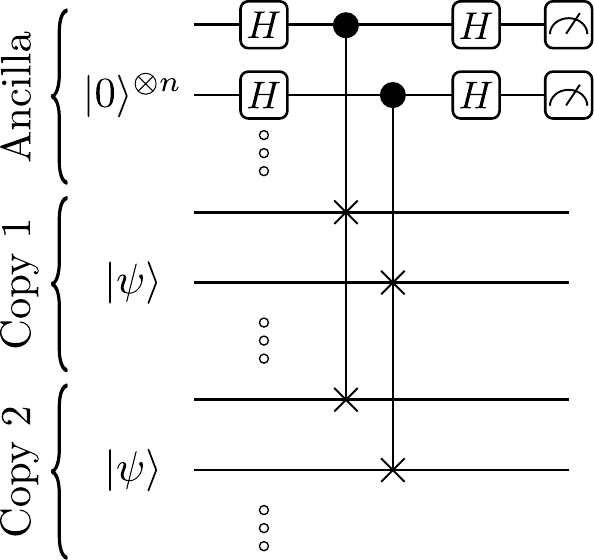}
         \caption{\textbf{Parallelized controlled SWAP test for measuring CE.} After preparing two copies of the state of interest, one performs a cSWAP test on each triplet of qubits from the ancilla and the two copies of $\ket{\psi}$. The resulting probabilities of bitstrings on the ancilla registers yield the CE.
         }
         \label{fig:swap}
     \end{figure}
\textit{A Hierarchy of Multi-qubit Product Structures --}
  We begin by denoting the maximum value of CE on $n$ qubits as
\begin{equation}
\label{Eq:CEmax1}
    \Cmax{n} = \max_{\ket{\psi}\in\mathbb{C}_2^{\otimes n}}\CE{\psi}.
\end{equation}
 Note that there is always a state achieving the maximal CE, which follows from Weierstrass' Extreme Value Theorem. To use this quantity for detecting different product state structures, we need to characterize how CE behaves for biseparable states.  This is given by the following.
\begin{proposition}\label{prop:bisep}
        For any biseparable state $\ket{\psi}=\ket{\psi_A}\ket{\psi_B}$, 
        \begin{equation}\label{eq:ce_bisep}
            \CE{\psi} = \CE{\psi_A} + \CE{\psi_B} - \CE{\psi_A}\CE{\psi_B}.
        \end{equation}
    \end{proposition}
\noindent From Proposition \ref{prop:bisep} it follows that if $|A|=k$ and $|B|=n-k$, then
\begin{equation}\label{eq:5}
        \mathcal{C}(\ket{\psi_A}\ket{\psi_B}) \leq \Cmax{k}+\Cmax{n-k}-\Cmax{k}\Cmax{n-k}.
    \end{equation}
Consequently, if a state has large enough CE, we know that it cannot be written as a product of pure states on $k$ and $n-k$ qubits.  By proceeding iteratively, one can obtain inequalities similar to that in  Eq.~\eqref{eq:5} for more separable cuts.  For any product state structure, one thus obtains a bound $\zeta^*$ on CE, and if $\CE{\psi} > \zeta^*$, then $\ket{\psi}$ defies any of the product structures bounded by $\zeta^*$ (e.g., see Fig. \ref{table:5hier}). Furthermore, by finding the largest CE possible across all bipartations, we obtain a threshold above which GME is certified,
    \begin{equation}
        \zeta(n)  = \max_{1\leq k<n}\Cmax{k}+\Cmax{n-k}-\Cmax{k}\Cmax{n-k}.
    \end{equation}

To find $\zeta^*$ and $\zeta(n)$ we must know $\Cmax{n}$ for arbitrary $n$.  Before discussing the optimization problem of Eq. \eqref{Eq:CEmax1}, we  observe that a simple upper bound on $\mc{C}^*(n)$ follows from assuming that all reduced density matrices are maximally mixed: 
\begin{equation}
\label{Eq:AME-bound}
\Cmax{n} \leq 1 - \frac{1}{2^n}\sum_{0\leq k \leq n}\binom{n}{k}2^{-\min(k,n-k)}.
\end{equation}
We note that a bound like this was considered in a similar task of verify GME in qudit states \cite{Qi2016}. Numerics from a recent work suggest that graph states come close to saturating this upper bound \cite{cullen2022calculating}. However, the bound in Eq. \eqref{Eq:AME-bound} is generally not tight. Indeed, pure states for which all bipartitions have maximally mixed marginals are called absolutely maximally entangled (AME), and AME states exist only for two, three, five, and six qubits \cite{Huber2017, Scott2004,Arnaud2013}.  Hence, this bound will necessarily be loose for any other number of qubits.

To obtain better bounds on $\Cmax{n}$, we first note that Eq. \eqref{Eq:CEmax1} can be equivalently expressed as
\begin{equation}
\label{Eq:CEmax2}
    \mc{C}^*(n)=1-\min_{\ket{\psi}\in\mbb{C}_2^{\otimes n}}\tr[M\op{\psi}{\psi}^{\otimes 2}],
\end{equation}
where $M=\otimes_{i=1}^n\Pi_+^{(i)}$, and with $\Pi_+^{(i)}$ the projector onto the symmetric subspace for the i\textsuperscript{th} qubit and its copy. As shown in the SM, we relax the optimization problem in Eq.~\eqref{Eq:CEmax2} by replacing the product state $\op{\psi}{\psi}^{\otimes 2}$ by an operator $X$ having a positive partial transpose. By further exploiting the symmetry in the objective function of~\eqref{Eq:CEmax2}, we can reduce the positive partial transpose relaxation to the following linear program (LP):
\begin{align}\label{eq:LP}
    1- \min & \ 3^ny_0\notag\\
    \text{subject to}\;\;& (i)\ \textbf{y} \geq 0\,,\notag\\
    &(ii)\ K\mathbf{y} \geq 0\,,\notag\\
    &(iii)\ \sum_{i=0}^{\lfloor \frac{n}{2} \rfloor} y_w3^{n-w}=1\,,
\end{align}
where $\mathbf{y}=(y_0,y_1,\cdots,y_{\lfloor \frac{n}{2} \rfloor})^T$ and $K$ is an $(n+1) \times (\lfloor \frac{n}{2} \rfloor + 1)$ matrix whose elements are quaternary Krawtchouk polynomials \cite{roman1992coding}. The Krawtchouk polynomials play an important role in classical coding theory, and we return to this curious connection latter in the paper.

\begin{table}[t]
    \centering
    \begin{tabular}{c|c|c}
    Number of Qubits & $\Cmax{n}$ & $\zeta(n)$\\
    \hline
     2 & 0.25 & 0\\
     3 & 0.375 & 0.25\\
     4 & 0.5 & 0.4375\\
     5 & 0.625 & 0.53125\\
     6 & 0.71875 & 0.625\\
     7 & $0.779296875^*$ & 0.71875\\
     8 & 0.828125 & 0.7890625\\
     9 & 0.8671875 & 0.83447265625\\
     10 & 0.8984375 & 0.87109375\\
     11 & 0.923828125 & 0.900390625\\
     12 & 0.94287109375 & 0.923828125\\
    \end{tabular}
    \caption{\textbf{Maximal values of CE and thresholds for detecting GME as a function of system size.} All $\Cmax{n}$ are achievable except for $n=7$. In this case, values of up to 0.7739 have been found numerically. $\zeta(n)$ is the threshold value above which the state must be GME.
    }
    \label{table:ubs_ths}
\end{table}

We have numerically solved this LP for up to $n=31$ (see SM).  While this just provides an upper bound on $\mc{C}^*(n)$, we have found states, corresponding to extremal quantum error correcting codes \cite{Rains98, Scott2004}, up to $n=18$ qubits achieving this upper bound, except for $n\in\{7,13,15,16\}$. We suspect that the results of our LP yield the exact maximum of $\mc{C}^*(n)$ for other system sizes as well.

Table~\ref{table:ubs_ths} presents our results for up to $12$ qubits.  For $n\notin \{2,3,5,6,7\}$, the value $\mc{C}^*(n)$ lies below the CE of an AME state, thereby providing an alternative proof for the nonexistence of AME states on systems with these numbers of qubits.  Table~\ref{table:5hier} provides an example  of our CE-based hierarchy for $n=5$.  Here, there are seven classes of product state structures, and we compute exact values for the maximal CE obtainable within each class.  Similar tables for up to $n=12$ can be found in the SM.

\textit{Random States --} Knowing the maximal possible values of CE, it is natural to ask how CE is distributed for random states and whether most GME states achieve $\CE{\psi}>\zeta(n)$. If the latter inequality holds with high probability, then the CE provides a good entanglement measure for certifying GME. In fact, we find that the average CE quickly goes to one and that most states lie above the GME threshold. Specifically, under the Haar measure, CE is distributed as follows.
\begin{proposition}\label{prop:haar}
    The Haar average of concentratable entanglement is
    \begin{equation}\label{eq:haar_av}
        \langle \mathcal{C} \rangle_{\text{Haar}} = 1- 2\frac{3^n}{4^n+2^n},
    \end{equation}
    with variance $\Var(\mathcal{C})_{\text{Haar}} =O\left((\frac{3}{16})^n\right)$. 
    \end{proposition}
We give an exact expression for the Haar variance in the SM. Note that Eq.~\eqref{eq:haar_av} goes to $1$ like $O(1-(\frac{3}{4})^n)$. In fact, in the appendix we establish the bound $\zeta(n)\leq \mathcal{C}^*(n)\leq 1-(\frac{3}{4})^n$ which matches the scaling of $\langle \mathcal{C} \rangle_{\text{Haar}}$.  Interestingly, for as few as 5 qubits the majority of states have CE greater than that of the GHZ state as $\CE{\text{GHZ}} = \frac{1}{2}-\frac{1}{2^n}$.   As non-GME states are measure zero under the Haar measure \cite{Lockhart2002}, an upper bound on the probability that a GME state cannot be verified via its CE can be easily found by applying Chebyshev's inequality. In Fig.~\ref{fig:scaling}(b) we demonstrate that this bound seemingly goes to 0 exponentially quickly.

We propose that a sort of hybrid method may be the best approach to verifying GME in practice. In the course of determining $\CE{\psi}$ one measures a set of bitstrings $\{\textbf{z}_i\}$. It turns out that each of these yields nontrivial information about the state.

\begin{table}
    \centering
    \begin{tabular}{c|c}
        Structure & Max CE $\zeta^*$\\
        \hline
        $5$ & 0.625\\
        $3\otimes 2$ & 0.53125\\
        $4\otimes 1$ & 0.5\\
        $2\otimes 2\otimes 1$ & 0.4375\\
        $3\otimes 1\otimes 1$ & 0.375\\
        $2\otimes 1\otimes 1\otimes 1$ & 0.25\\
        $1\otimes 1\otimes 1\otimes 1\otimes 1$ & 0\\
    \end{tabular}
    \caption{\textbf{Hierarchy for 5 qubit states based on CE.} Notation such as $2\otimes 2\otimes 1$ denotes the class of states having the form $\ket{\psi_{AB}}\ket{\psi_{CD}}\ket{\psi_E}$ for any labeling of parties.  A state cannot have a product state structure according to any partitioning lying lower in the hierarchy than its CE.}
        \label{table:5hier}
\end{table}

\begin{proposition}
    Given a measurement of a bitstring $\textbf{z}$ in the cSWAP test, the state cannot be a product with respect to any partition such that the Hamming weight of the substring of $\textbf{z}$ restricted to any set in the partition is odd.
\end{proposition}
\begin{corollary}
    A measurement of bitstring $\textbf{z} \neq \textbf{0}$ implies that $\ket{\psi}$ cannot be a product state with respect to half of all possible bipartitions. Further, measuring $k$ linearly independent bitstrings implies that a state could only be biseparable with respect to $2^{n-1-k}$ bipartitions.
\end{corollary}
For example, say one measures the bitstring $1100$ then it is not possible for the state to be biseparable with respect to the partitions: $\{1\} : \{2,3,4\}$, $\{2\} : \{1,3,4\}$, $\{1,3\} : \{2,4\}$, or $\{1,4\} : \{2,3\}$. While verifying GME may require exponential shots, just a handful of bitstrings removes most possibilities. In practice one could do the following: run cSWAP to determine $\CE{\psi}$ up to some resolution, in the process measuring a set of bitstrings, and then proceed to verify entanglement between the remaining possible $2^{n-1-k}$ bipartitions.

\begin{figure}[t]
    \centering
    \includegraphics[width=0.9\columnwidth]{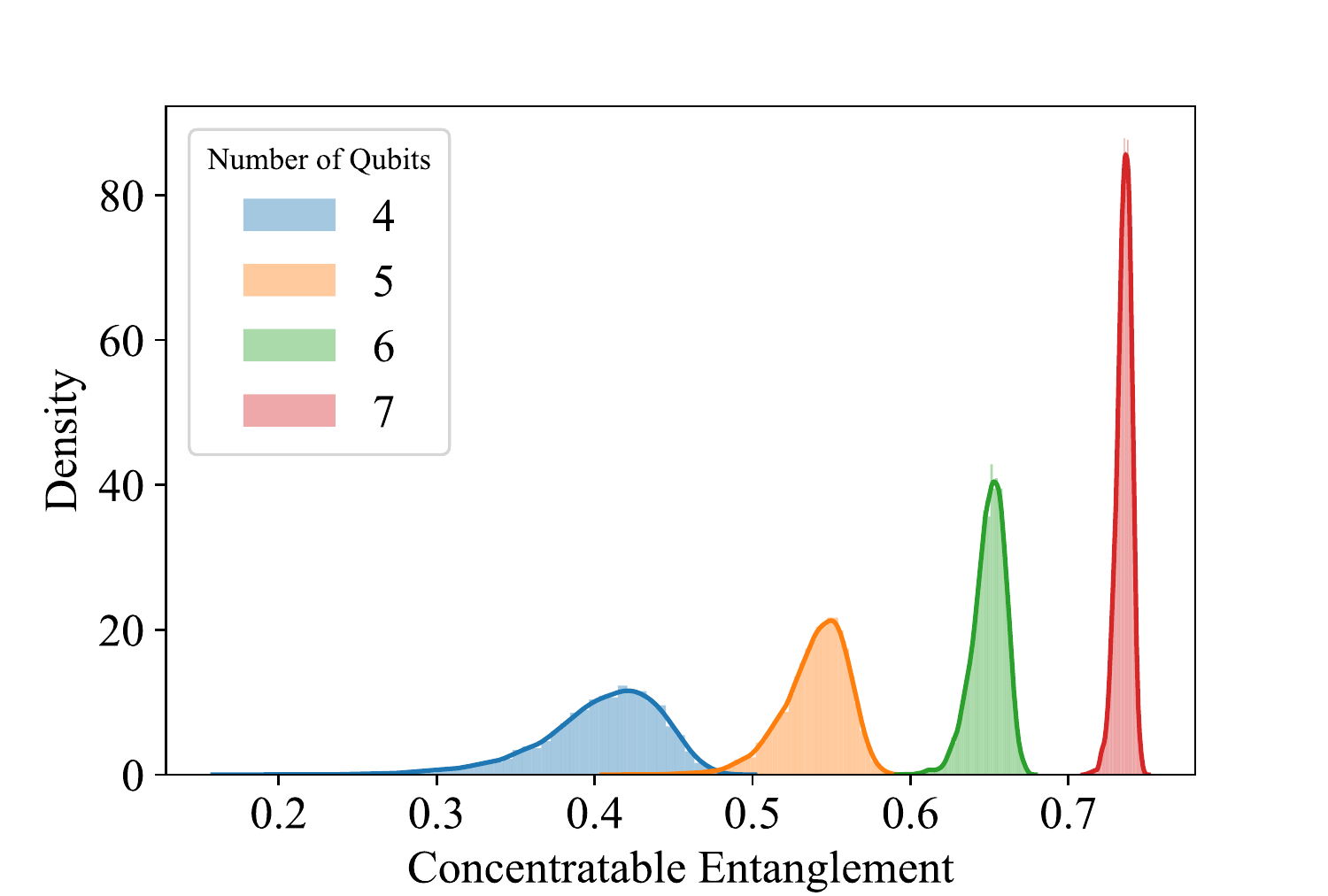}
    \caption{\textbf{Haar distribution of concentratable entanglement for 4 to 7 qubits.} As expected from Eq.~\ref{eq:haar_av}, the mean grows with the system size while the distribution grows tighter. Data comes from 6000 random samples for each system size.}
    \label{fig:haar_dist}
\end{figure}

\textit{Applications.  Connections to Other Measures and Entanglement Concentration --}
We next consider what other properties of entanglement we can learn by studying the CE. In \cite{Beckey2021} it is shown that the entanglement measures of Refs. \cite{Wong2001, Carvalho2004} correspond to special cases of concentratable entanglement. Here we provide connections to two other entanglement measures: the tensor rank/Schmidt measure \cite{Eisert2001} and the average linear entropy \cite{Meyer2002,Brennen2003} (additionally, a link to geometric measure of entanglement \cite{Wei2003} is given in the SM).

Every multipartite state can be written as a (non-unique) sum of product states $\ket{\psi} = \sum_{i=1}^{r}c_i(\bigotimes_{j=1}^n\ket{\phi_i^{(j)}})$. Viewing $\ket{\psi}$ as a tensor, this is a canonical polyadic (CP) decomposition of $\ket{\psi}$ \cite{hitchcock1927expression, harshman1970foundations}. The minimum $r$ for which such a decomposition exists is called the CP rank, or tensor rank, of $\ket{\psi}$ and is denoted as $\text{rk}(\ket{\psi})$.  The CP rank can be used to quantify entanglement \cite{Brylinski2002, Eisert2001}, and it has an operational meaning in terms of stochastic LOCC convertibility \cite{Chitambar2008, Chen2010}.  Classifying entanglement in terms of CP rank is fundamentally different than categorizing entanglement as being either GME or  non-GME.  For example, rank distinguishes the two inequivalent three-qubit GME classes (GHZ and W types) \cite{Dur2000}. Further, no reduced density matrix can have matrix rank greater than the CP rank. However, since $\Tr[\rho^2] \geq \frac{1}{d}$, where $d$ is the matrix rank of $\rho$, CP rank provides a lower bound on purities. For a subsystem $\alpha$ of a state of rank $\rank(\ket{\psi}) = R$, we have that $\Tr[\rho_\alpha^2] \geq \max(\frac{1}{R}, 2^{-\min(k,n-k)})$. As an example, a state on 5 qubits of CP rank 3 must have CE less than $1-\frac{1}{32}[1+\binom{5}{2}\frac{1}{2}+\binom{5}{3}\frac{1}{3}+\binom{5}{4}\frac{1}{2}+1]=\frac{55}{96}\approx 0.573 < 0.625 = \Cmax{5}$. We can readily generalize this:

\begin{proposition}\label{prop:cp}
    For a state such that $\text{rk}(\ket{\psi})=R,$
    \begin{equation}
    \label{Eq:CP-bound}
        \CE{\psi} \leq 1-2^{-n}[\sum_{k}\binom{n}{k}\max(\frac{1}{R},2^{-\min(k,n-k)})]
    \end{equation}
\end{proposition}
\noindent Consequently, one can guarantee that $\text{rk}(\ket{\psi})>R$ if $\ket{\psi}$ has CE greater than the bound in Eq. \eqref{Eq:CP-bound}. This allows us to rule out types of entanglement via cSWAP tests. For example, measuring $\CE{\psi}>\frac{1}{2}-\frac{1}{2^n}$ would imply that $\ket{\psi}$ is not a GHZ-type state (i.e., rank two). Furthermore, demonstrating high CP rank has connections to simulatability \cite{Ma2022}. We note here the asymptotic value of Proposition~\ref{prop:cp}:
\begin{corollary}
    As $n\rightarrow\infty$, a state of CP rank R on n qubits has concentratable entanglement less than $\frac{R-1}{R}$.
\end{corollary}

\begin{figure}[]
    \centering
    \includegraphics[width=1\columnwidth]{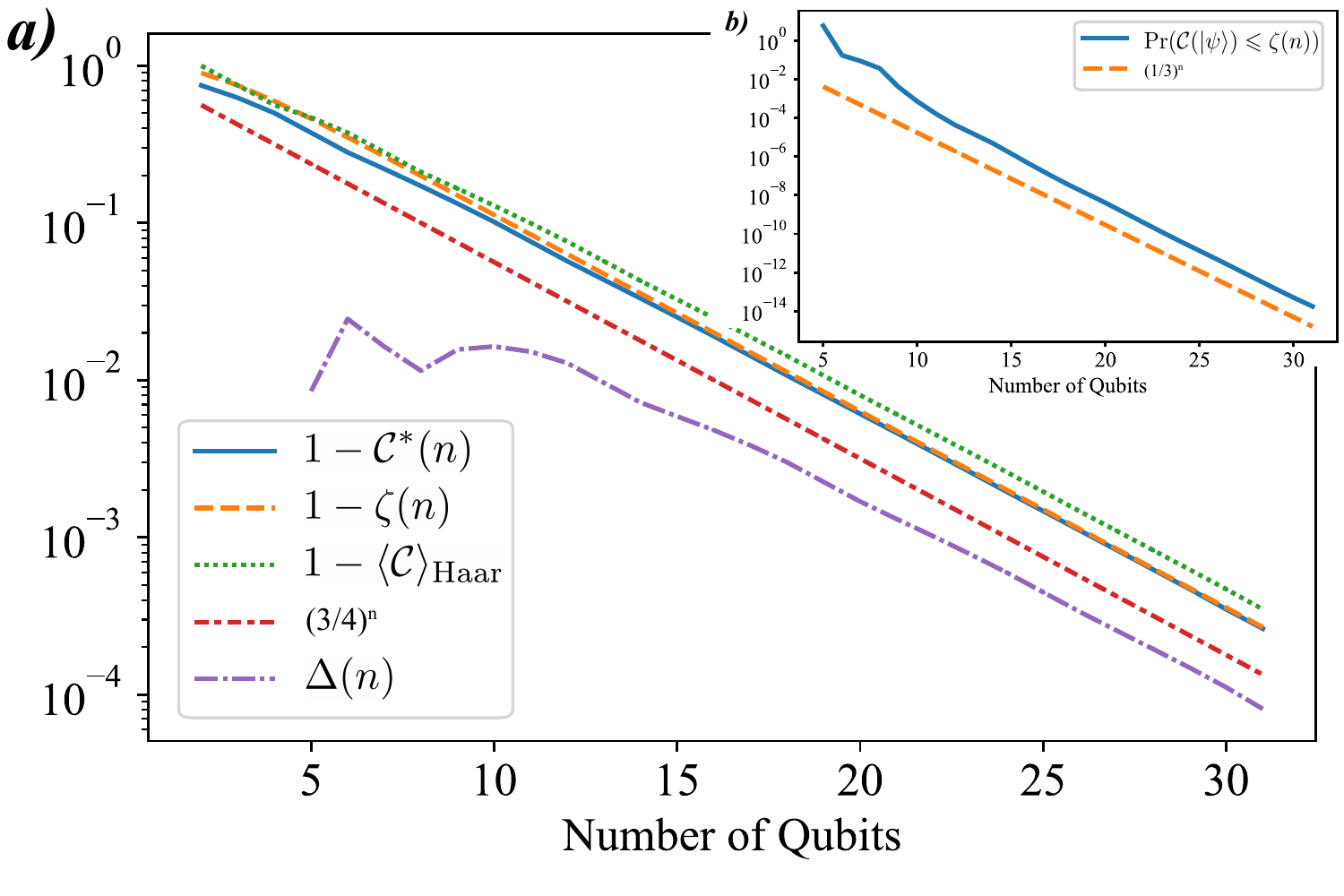}
    \caption{\textbf{Scaling of properties of concentratable entanglement.} \textbf{(a)} The maximal/average values and the GME threshold go to 1 like $O((\frac{3}{4})^2)$, yet $\zeta(n)$ does not grow larger than $\Cmax{n}$ or $\langle C \rangle_{\text{Haar}}$. \textbf{(b)} Using Chebyshev's inequality, we upper bound the probability that a random state has CE below the GME threshold. This seemingly goes to 0 like $O(3^{-n})$.}
    \label{fig:scaling}
\end{figure}

Another entanglement measure we can relate to CE is the average linear entropy, $1-\frac{1}{n}\sum_{i}\tr[\rho_i^2]$.  It turns out that this connection is intimately related to the operational fact that cSWAP provides a form of universal distortion-free entanglement concentration \cite{matsumoto2007universal}. Whenever a failure (i.e., outcome $1$) is measured on the ancilla system in a cSWAP test, one knows that a Bell pair must have been created on the two system qubits. This allows for the distribution of Bell pairs in adversarial scenarios. Say Alice and Bob are operating in separate labs with some trusted node through which they can enact SWAP operations. Based on ancillas in Alice's lab, they can perform cSWAP. Even if Bob acts maliciously, Alice knows with certainty that they now share a Bell pair given measurement outcome $1$.

Thus, it is interesting to ask how well a given state $\ket{\psi}$ performs at concentrating Bell pairs using the cSWAP method. While CE corresponds to the probability of creating any Bell pairs whatsoever, it may be informative to consider other performance metrics such as the expected number of Bell pairs generated. We show that this value is precisely the average linear entropy.
\begin{proposition}
    Given an n-qubit state $\ket{\psi}$, the expected number of Bell pairs from running the parallelized SWAP test on $\ket{\psi}^{\otimes 2}$ is 
    \begin{equation}
        \mathcal{B}(\ket{\psi}) = \frac{1}{2}(n-\sum_i \Tr[\rho_i^2]),
    \end{equation}
    with variance
    \begin{equation}
        \text{Var}(\mathcal{B}(\ket{\psi})) = \frac{n}{4}-\frac{\sum_i\Tr[\rho_i^2]}{4}(\sum_j\Tr[\rho_j^2]+2)+\sum_{i,j}\Tr[\rho_{i,j}^2].\notag
    \end{equation}
\end{proposition}

Recall that a state is called $k$-uniform if all reduced density matrices consisting of $k$ or fewer parties are maximally mixed.  We thus have the following.
\begin{corollary}
\label{Cor:avg-bell-pair}
    The maximum number of expected Bell pairs using the parallelized SWAP test is $\frac{n}{4}$, and it is achieved by any $1$-uniform state. Further, any $2$-uniform state achieves the maximal expected number while minimizing the variance.
\end{corollary}
\noindent While writing this manuscript it was brought to our attention that the expected value formula was independently in \cite{Brennen2003} but without an explicit proof. For completion we give a detailed proof in the SM. 

\textit{Connections to Coding Theory --}  While not obvious at first, our work hints at deep connections between CE and quantum coding theory. Previous works have related error correction to similar entanglement measures \cite{Scott2004} and AME states \cite{Huber_2018,Raissi_2018}, but the connection to CE is novel. In this work we identify new connections involving the states achieving $\Cmax{n}$ and the LP of Eq. \eqref{eq:LP}.

To briefly review (see \cite{nielsen_chuang_2010,gottesmanthesis,Steane963} for more details), a quantum error correcting code is a subspace of an $n$-qubit Hilbert space with corresponding projector $P_Q$. To encode information we apply a map from a $k$-qubit space into this subspace. The code is constructed such that certain errors can be detected/corrected. We say that the weight of an error operator $E$ is the number of parties it acts nontrivially on. For example, $X\id Z$ has weight 2. A code has distance $d$ if all errors of weight less than $d$ are detectable. Such a code is denoted by $(n,k,d)$.

Among all codes, stabilizer codes \cite{nielsen_chuang_2010,gottesmanthesis} are ubiquitous. Here the codespace is the joint eigenspace of an abelian subgroup of the Pauli group. If this subgroup has $n-k$ generators the codespace is of dimension $2^k$, and we denote the code by $[[n,k,d]]$. These codes have a natural correspondence with codes over GF(4) \cite{Calderbank98} and they can be represented as a set of strings on GF(4). A code is said to be self-dual if it is equal to its orthogonal complement (over GF(4)). Note that $[[n,0,d]]$ codes (pure states) are self-dual. Self-dual stabilizer codes are type II if all codewords have even weight, and type I otherwise. Type-dependent bounds on distance can be found in \cite{Scott2004,Rains98}. A code achieving these bounds is said to be extremal. We list these bounds below for clarity.
\begin{equation}
    d \leq \begin{cases}
        2\lfloor \frac{n}{6} \rfloor +1 & \text{type I, }n\% 6 = 0\\
         2\lfloor \frac{n}{6} \rfloor +3 & \text{type I, }n\% 6 = 5\\
          2\lfloor \frac{n}{6} \rfloor +2 & \text{type I, }n\% 6 \notin \{0,5\}\\
           2\lfloor \frac{n}{6} \rfloor +2 & \text{type II}
    \end{cases}\,.
\end{equation}

The first connection between CE and error correction is empirical.  Based on numerical searching and constructing solutions to the LP in Eq. \eqref{eq:LP}, we found that all states maximizing CE correspond to extremal codes (see SM).  We thus conjecture that if a $[[n,0,d]]$ stabilizer code achieves the maximal concentratable entanglement on $n$ qubits it must be an extremal code. Note that the converse direction does not hold. For example, the 10 qubit extremal code in Table 1 of \cite{Scott2004} does not achieve $\Cmax{n}$.

We now consider another property of stabilizer codes known as enumerators. Classically, these arise in linear programming bounds on codes \cite{Roman1992}. In analogy to the weights of classical codes, quantum codes can be described through their Shor-Laflamme enumerators \cite{Shor1997}:
\begin{equation}
    A_i  = \frac{1}{k^2}\sum_{\sigma: w(\sigma) = i}\Tr[\sigma P_Q]\Tr[\sigma^\dagger P_Q],
\end{equation}
\begin{equation}
    B_i  = \frac{1}{k}\sum_{\sigma: w(\sigma)=i}\Tr[\sigma P_Q \sigma^\dagger P_Q],
\end{equation}
where $\sigma$ is in the Pauli group and $w(\sigma)$ is the corresponding weight. The quantum MacWilliams identity \cite{Shor1997} connects these: $B_i = 2^{k-n}\sum_{j}K_i(j)A_j$, where $K_i(j)$ are again the four-ary Krawtchouk polynomials. This is identical to a constraint in the LP for maximizing CE (see Eq. \eqref{eq:LP}). As any enumerator $\{A_i\}$ is in the feasible set (up to multiplicative factors), this LP also yields a bound on quantum codes. Denoting the solution to the LP in Eq. \eqref{eq:LP} by $L(n)$ we find the following.
\begin{proposition}
   A stabilizer code such that $A_{2i+1}=0$ must satisfy the inequality
    \begin{equation}
        B_n \leq \frac{1}{L(n)}\frac{3^n}{2^{n-k}}.
    \end{equation}
\end{proposition}
\noindent Note that when the LP yields an achievable value of CE, this can be rewritten as $B_n \leq \frac{1}{1-\Cmax{n}}\frac{3^n}{2^{n-k}}$.

Now consider the problem of finding the maximal value of the expected number of Bell pairs. While in Corollary \ref{Cor:avg-bell-pair} we found this value to be $\frac{n}{4}$, we could, in a similar manner to CE, solve this problem via linear programming. This yields another bound on codes.
\begin{proposition}
    A stabilizer code such that $A_{2i+1}=0$ must satisfy the inequality
    \begin{equation}
        \sum_{i=0}^n i\cdot 3^{-i}\cdot A_i \leq \frac{n}{4}\frac{2^{n-k}}{3^n}B_n.
    \end{equation}
\end{proposition}
If the code is self-dual then $A_i = B_i$. Thus, we recover the fact that type II codes do not exist for odd $n$.  Combining these bounds yields $\sum_{i=0}^n i\cdot 3^{-i}\cdot A_i \leq \frac{n}{4L(n)}$. Both of these bounds are special cases of more general inequalities holding for any stabilizer code, which can be found in the SM.

\textit{Discussion --} In this work we showed that the CE and the parallelized SWAP test can be used to determine various multipartite entanglement properties of pure states. Particularly, CE induces a hierarchy upon pure states, separating genuine multipartite entangled states from biseparable. For even modest system sizes, most states fall neatly into these hierarchies and thus CE can verify GME for most states.

While our analysis has focused on pure states, the hierarchy has some robustness for certifying different entangled structures for mixed state.  The key relation needed is a generalization of Proposition \ref{prop:bisep}.  As shown in the SM, every $n$-qubit product density matrix $\rho=\rho^A\otimes \rho^B$ with $|A|=k$ and $|B|=n-k$ satisfies the bound
\begin{equation}
    \mc{C}(\rho)\leq \mc{C}^*(k)+\mc{C}^*(n-k)- \mc{C}^*(k)\mc{C}^*(n-k)+2\sqrt{S_L(\rho)},\notag
\end{equation}
where $S_L(\rho)=1-\tr[\rho^2]$ is the linear entropy of $\rho$.  Here, $\mc{C}(\rho)$ is being defined just as in Eq. \eqref{eq:CE}.  Consequently, entanglement structures can still be verified using the value $\zeta^*$ provided $\rho$ has sufficiently high purity.

\begin{acknowledgments}
We are grateful to Olgcia Milenkovic for helpful discussions on classical and quantum error correction.  We would also like to thank Ian George and Brian Doolittle for fruitful discussions in deriving and implementing the LP for upper bounding $\Cmax{n}$. L.S. and E.C. acknowledge support from the NSF Quantum Leap Challenge Institute for Hybrid Quantum Architectures and Networks (NSF Award 2016136). L.S. and M.C. were intially supported by
ASC Beyond Moore’s Law project at Los Alamos National Laboratory (LANL). MC acknowledge support by NSEC Quantum Sensing at LANL. This work was also supported by the Quantum Science Center (QSC), a National Quantum Information Science Research Center of the U.S. Department of Energy (DOE).
\end{acknowledgments}

\bibliography{ref.bib,supp_ref.bib}

\newpage
\onecolumngrid

\input{supp.tex}

\end{document}

%% file: supp.tex
\appendix


\vspace{0.25in}
\begin{centering}
\large\textbf{Supplementary Material for: ``\textit{A Hierarchy of Multipartite Correlations Based on Concentratable Entanglement}''}
\end{centering}

\vspace{-0.25in}
    
    
    
    


\tableofcontents
\vspace{0.25in}

In this Supplemental Material we provide further details for the manuscript ``\textit{A Hierarchy of Multipartite Correlations Based on Concentratable Entanglement}'', including detailed proofs and examples of the propositions and concepts therein. 

The Supplemental Material is organized as follows. In Section~\ref{CE_rev} we review the concept of concentratable entanglement (CE), which was introduced in Ref.~\cite{Beckey2021}. In Sections~\ref{sec;LP}--\ref{sect:othermeasures} we sequentially offer proofs of the propositions in the main text, with several additional corollaries. The first few proofs are related to finding $\Cmax{n} = \max_{\ket{\psi}}\CE{\psi}$ and verifying genuine multipartite entanglement (GME). Next, in Section~\ref{sect:othermeasures} we lay out connections to other entanglement measures, namely CP rank and geometric measure of entanglement. Then, in Section~\ref{sec:bellpairs}, we discuss the parallelized controlled SWAP test as a form of entanglement generation and find that the SWAP test yields another entanglement measure. In Section~\ref{sec:coding} we connection CE to coding theory.  Numerical results from our linear programming bound are given in Section~\ref{data}.  Finally, in Section~\ref{Sect:robustness} we give a continuity bound on CE, which allows applicability of our hierarchy on mixed states.

\section{Review of Concentratable Entanglement}\label{CE_rev}
Here we provide a brief overview of the CE, which was introduced in \cite{Beckey2021}. We refer the reader to the corresponding manuscript for a more detailed discussion of the concept.

Denote by $\mathcal{S}=\{1,2,\ldots,n\}$ the set of labels for the qubits in a system. For a set of labels $s\in Q(\mathcal{S})\backslash\{\emptyset\}$ (where $Q$ is the powerset of $\mathcal{S}$), the CE is defined as
\begin{align}
    \mathcal{C}_{\ket{\psi}}(s) = 1-\frac{1}{2^{\lvert s \rvert}}\sum_{\alpha\in Q}\Tr[\rho_\alpha^2].
\end{align}
Note that in this work we are primarily concerned with the case where $s=\mathcal{S}$, and for this special case we denote the CE as $\CE{\psi}$. One of the most attractive properties of the CE is that it can be measured via the controlled-SWAP test: given two copies of a state and an ancilla initialized to $\ket{\textbf{0}}$, apply Hadamard gates to each ancilla, then controlled SWAPs between the $k$-th ancilla and the $k$-th qubit in each copy of the states, and finally Hadamard gates on each ancilla again. This is illustrated in Fig~\ref{fig:swap_test}. By measuring the ancilla registers in the computational basis, one finds the CE through
\begin{align}
    \mathcal{C}_{\ket{\psi}}(s) = 1-\sum_{z\in\mathcal{Z}_0(s)} p(z),
\end{align}
where $\mathcal{Z}_0(s)$ is the set of bitstrings with 0's on the indices in s. In the case that $s=\mathcal{S}$, this is simply $1-p(\textbf{0})$. Further, in \cite{Beckey2021}, it is noted that, for systems of even numbers of qubits, $p(\textbf{1}) = \tau_{(n)}/2^n$, the $n$-tangle introduced in \cite{Wong2001}. Also, note that a measurement of $\ket{1}$ on an ancilla register leaves the $k$-th qubits in the copies of $\ket{\psi}$ in the state $\ket{\Phi^-}=\frac{1}{\sqrt{2}}(\ket{01}-\ket{10})$. Thus, $1-p(\textbf{0})$ corresponds to the probability of creating at least one Bell pair using $\ket{\psi}^{\otimes 2}$ and the parallelized controlled-SWAP test.

For the purposes of this work it will be convenient to consider a linearized reformulation of CE. If we measure 0 on the $i$-th ancilla then the Kraus operator acting on $A_i$ and $B_i$ is $\mathcal{K} = \frac{1}{2}(\id+\mathbb{F}^{(A_iB_i)}) = \Pi_+^{(k)}$, where $\mathbb{F}$ is the swap operator. That is, if we measure 0 we know that the symmetric projector has been applied accross the corresponding qubits in the two copies of $\ket{\psi}$. Similarly, one can show that a measurement of 1 corresponds to $\mathcal{K} = \frac{1}{2}(\id-\mathbb{F}^{(A_iB_i}) = \Pi_-^{(k)}$, the antisymmetric projector. Thus, the probability of measuring a bitstring $\textbf{z}$ is 
\begin{align}
    p(\textbf{z}) = \Tr[\otimes_{i}\frac{1}{2}(\id+(-1)^{z_i}\mathbb{F})^{(i)}\dya{\psi}^{\otimes 2}]. 
\end{align}
   We can then write CE as
\begin{align}
    \CE{\psi} = 1-\Tr[M\dya{\psi}^{\otimes 2}],
\end{align}
where we are defining $M := \otimes_i \Pi_+^{(i)}$. This will be essential in our proofs.

\begin{figure}[t]
    \centering
    \includegraphics{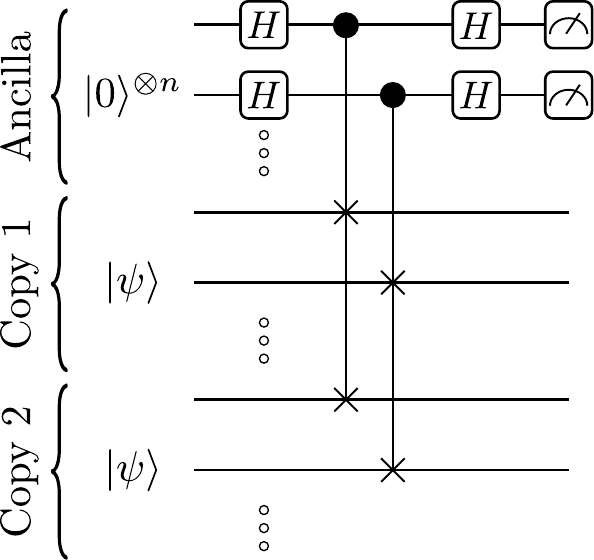}
    \caption{\textbf{Parallelized controlled SWAP test for measuring CE.} After preparing two copies of a state, simply perform a controlled SWAP test on each triplet of qubits from the ancilla and the two copies of $\ket{\psi}$. The resulting probabilities of bit strings on the ancilla registers yield the CE.}
    \label{fig:swap_test}
\end{figure}

\section{Proof of Propositions 1 and Corollary}\label{hier_proof}
First, we prove a simple formula for the CE of biseparable states.
\begin{proposition}\label{prop:SM_bisep}
        For any biseparable state $\ket{\psi}=\ket{\psi_A}\ket{\psi_B}$, 
        \begin{align}\label{eq:SM_ce_bisep}
            \CE{\psi} &= \CE{\psi_A} + \CE{\psi_B} - \CE{\psi_A}\CE{\psi_B}.
        \end{align}
\end{proposition}
Instead of proving this directly, we consider the more general case of $\mathcal{C}_{\ket{\psi}}(s)$, of which the desired expression is a special case.
\begin{proposition*}
    For a biseperable state $\ket{\psi}=\ket{\psi_A}\ket{\psi_B}$, 
    \begin{align}\label{eq:ces_bisep}
            \mathcal{C}_{\ket{\phi}}(s) &= \mathcal{C}_{\ket{\phi_A}}(s) + \mathcal{C}_{\ket{\phi_B}}(s) - \mathcal{C}_{\ket{\phi_A}}(s)\mathcal{C}_{\ket{\phi_B}}(s).
        \end{align}
\end{proposition*}
\begin{proof}
    \begin{align}
    \mathcal{C}_{\ket{\phi}}(s) & = 1-\frac{1}{2^{|s|}}\sum_{\alpha\in\mathcal{P}(s)}\Tr[\rho_\alpha^2]\\
    & = 1-\frac{1}{2^{|s|}}\sum_{\alpha\in\mathcal{P}(s)}\Tr[\rho_{\alpha\cap A}^2\otimes\rho_{\alpha\cap B}^2]\\
    & = 1-\frac{1}{2^{|s|}}\sum_{\alpha\in\mathcal{P}(s)}\Tr[\rho_{\alpha\cap A}^2]\Tr[\rho_{\alpha\cap B}^2]\\
    & = 1-\frac{1}{2^{|s|}}\sum_{\alpha_a\in\mathcal{P}(s\cap A)}\Tr[\rho_{\alpha_a}^2]\sum_{\alpha_b\in\mathcal{P}(s\cap B)}\Tr[\rho_{\alpha_b}^2]\\
    & = 1-\frac{1}{2^{|s|}}2^{|s\cap A|}(1-\mathcal{C}_{\ket{\psi_A}}(s\cap A))2^{|s\cap B|}(1-\mathcal{C}_{\ket{\psi_B}}(s\cap B))\\
    & = 1-(1-\mathcal{C}_{\ket{\psi_A}}(s\cap A))(1-\mathcal{C}_{\ket{\psi_B}}(s\cap B))\\
    & = \mathcal{C}_{\ket{\psi_A}}(s\cap A)+\mathcal{C}_{\ket{\psi_B}}(s\cap B)-\mathcal{C}_{\ket{\psi_A}}(s\cap A)\mathcal{C}_{\ket{\psi_B}}(s\cap B)
    \end{align}
    Where we have used the fact that $A\cup B = S$ and $A\cap B = \emptyset$. This implies that $\mathcal{P}(s)\cong\mathcal{P}(s\cap A)\times\mathcal{P}(s\cap B)$ and $|s|=|s\cap A|+|s\cap B|$.
\end{proof}
By taking $s=S$, we prove Prop.~\ref{prop:SM_bisep}. We note here that Prop.~\ref{prop:SM_bisep} allows us to find the CE of many copies of a state:
\begin{corollary*}
        The CE of k copies of a state $\ket{\psi}$ is
        \begin{align}
            \mathcal{C}(\ket{\psi}^{\otimes k}) & = 1-(1-\CE{\psi})^k.
        \end{align}
    \end{corollary*}
\begin{proof}
     We prove the claim via induction, with 2 copies as the base case. $\mathcal{C}(\ket{\psi}^{\otimes 2}) = 2\CE{\psi}-\CE{\psi}^2=1-(1-\CE{\psi})^2$Assume this holds true for $k$, and the claim holds for $k+1$:
    \begin{align}
        \mathcal{C}(\ket{\psi}^{\otimes k+1}) & = \CE{\psi}+\mathcal{C}(\ket{\psi}^{\otimes k})-\CE{\psi}\mathcal{C}(\ket{\psi}^{\otimes k})\\
        & = \CE{\psi}+(1-\CE{\psi})^k-\CE{\psi}(1-(1-\CE{\psi})^k)\\
        & = 1-(1-\CE{\psi})^{k+1}.
    \end{align}
\end{proof}

\section{Upper Bounds Through Mathematical Programming}\label{sec;LP}
To maximize the CE we want to minimize $\Tr[M\op{\psi}^2]$. While one would ideally solve this optimization problem over the set of symmetric bipartite pure states, this set is not convex and a direct optimization will be difficult. Instead, one can relax to the set of SWAP symmetric positive partial transpose states, which is convex. This yields the following semidefinite program (SDP):
\begin{align}\label{eq:SDP}
    \min_{X\in\mathbb{C}^{2^{2n}\times 2^{2n}}} &\Tr[MX]\\
    \text{subject to}\;\;& (i)\;\mathbb{F}^{AB}X=X\notag\\
    & (ii)\;X\mathbb{F}^{AB}=X\notag\\
    & (iii)\;X\geq 0\notag\\
    & (iv)\;\text{Tr}[X]=1\notag\\
    & (v)\; X^\dagger = X\notag\\
    & (vi)\; X^{\Gamma_B} \geq 0\notag.
\end{align}
While the optimization problem above can be solved, its complexity grows too quickly (as $O(4^n)$) to easily do so in practice. We will show that there exists an equivalent linear program  on $O(n)$ variables.

To convert Eq.~\eqref{eq:SDP} to a linear program (LP) first consider the Haar unitary twirl operation on a bipartite operator \cite{Werner1989}:
\begin{align}
    \mathcal{T}(X) & = \int (U\otimes U)X(U\otimes U)^\dagger dU\\
    & = \binom{d+1}{2}^{-1}\Tr[\Pi_+X]\Pi_++\binom{d}{2}^{-1}\Tr[\Pi_-X]\Pi_-,
\end{align}
where the integration is taken with respect to the Haar measure on $\mathcal{U}(d)$. For qubit states $d=2$ and the expression reduces to $\frac{\Tr[\Pi_+X]}{3}\Pi_++\Tr[\Pi_-X]\Pi_-$. As $\Tr[\Pi_+]=3$, it is clear that $\mathcal{T}(\Pi_+) = \Pi_+$. Similarly, $\mathcal{T}(\Pi_-) = \Pi_-$. Now consider the multipartite twirling $\mathcal{T}^{\otimes n}(X)$, which maps a bipartite state on two $n$-qubits systems to $\text{span}(\{\Pi_+,\Pi_-\}^{\otimes n})$. Via the invariance of $\Pi_+$, $\mathcal{T}^{\otimes n}(M)=M$. Further, $\mathcal{T}$ is self-adjoint, as can be readily proven:
\begin{align}
    \ipa{Y}{\mathcal{T}(X)} & = \Tr[Y^\dagger\int (U\otimes U)X(U\otimes U)^\dagger dU]\\
    & = \int \Tr[Y^\dagger(U\otimes U)X(U\otimes U)^\dagger dU]\\
    & = \int \Tr[(U\otimes U)^\dagger Y^\dagger (U\otimes U) X dU]\\
    & = \Tr[\mathcal{T}(Y)^\dagger X]\\
    & = \ipa{\mathcal{T}(Y)}{X}.
\end{align}
Using the invariance of M and the fact that $\mathcal{T}$ is self-adjoint, we see that
\begin{align}
    \Tr[MX] & = \Tr[\mathcal{T}^{\otimes n}(M)X]\\
    & = \Tr[M\mathcal{T}^{\otimes n}(X)].
\end{align}
Thus, without loss of generality we can restrict $X$ to be in $\text{span}(\{\Pi_+,\Pi_-\}^{\otimes n})$ and can write it in terms of $2^n$ coefficients:
\begin{align}
    X & = \sum_{\textbf{s}\in\mathbb{F}_2^n} x_{\textbf{s}}\pi_{\textbf{s}},
\end{align}
where $\pi_\textbf{s} = \bigotimes_k \frac{1}{2}(\mathbb{1}+(-1)^{s_k}\mathbb{F}^{(k)})$. We now note that $M$ is symmetric under permutations of $A_1B_1, A_2B_2,\ldots A_nB_n$. Thus, it is invariant under twirling with respect to the $n$-party symmetry group $S_n$:
\begin{align}
    \mathcal{T}_{S_n}(X) = \frac{1}{n!}\sum_{\sigma\in S_n}V_\sigma X V_\sigma^\dag.
\end{align}
Here we note that we are working with the natural permutation representation of the symmetry group. That is, $\sigma\in S_n$ is associated to a linear operator $V_\sigma$ such that $V_\sigma (\bigotimes_i \ket{a_i}) = \bigotimes_i \ket{a_{\sigma^{-1}(i)}}$. As $M = \bigotimes_{k=1}^n \Pi_+^{(k)}$, it is clear that $M$ is invariant under this twirling operation as well. As in the case of Haar twirling, this operation can also be easily proven to be self-adjoint. And thus, $\ipa{M}{X}=\ipa{\mathcal{T}_{S_n}(M)}{X} = \ipa{M}{\mathcal{T}_{S_n}(X)}$. From Haar twirling, we have seen that $X$ can be taken to be in $\text{span}\{\Pi_+,\Pi_-\}^{\otimes n}$. The symmetry twirl will take each $\pi_\textbf{s}$ to a symmetrized version:
\begin{align}
    \mathcal{T}_{S_n}(\pi_{\textbf{s}}) & = \frac{1}{n!}\sum_{\substack{\textbf{s'}\in\mathbb{F}_2^n\\ w(\textbf{s})=w(\textbf{s'})}}\pi_{\textbf{s'}}\\
    & =: \frac{1}{n!}\tau_{w(\textbf{s})},
\end{align}
where $w(\textbf{s})$ denotes the Hamming weight of bitstring $\textbf{s}$. Thus, we can now assume, without loss of generality, that $X=\sum_{w=0}^{n}x_w\tau_w$. As both twirls above map positive operators to positive operators, we require that $\mathcal{T}_{S_n}\circ \mathcal{T}^{\otimes n}(X)\geq 0$. Since $\{\tau_w\}$ is an orthogonal set of positive operators, we require $x_w \geq 0$ for all $w$. Next, we consider left and right swap invariance, i.e. $\mathbb{F}^{AB}X=X$ and $X\mathbb{F}^{AB}=X$. This constraint commutes with the action of $\mathcal{T}$ and $\mathcal{T}_{S_n}$ and thus we still require it. Note that $\mathbb{F}\Pi_+=\Pi_+$ and $\mathbb{F}\Pi_-=-\Pi_-$. Thus, $\mathbb{F}^{AB}\pi_{\textbf{s}}=(-1)^{w(\textbf{s})}\pi_{\textbf{s}}$ and $\mathbb{F}^{AB}\tau_w = (-1)^w\tau_w$. From this, it is clear that $x_w=0$ for all odd weights and we need only consider the $\lceil \frac{n+1}{2} \rceil$ terms corresponding to even weights. As $\mathbb{F}$ commutes with $\Pi_+$ and $\Pi_-$, left and right swap invariance yield the same feasible region. Both twirls are trace preserving, yielding the constraint $\Tr[\mathcal{T}_{S_n}\circ \mathcal{T}^{\otimes n}(X)]=1$. $\Tr[\Pi_+] = 3$ and $\Tr[\Pi_-]=1$ thus $\Tr[\tau_w]=\binom{n}{w}3^{n-w}$ and we require that $1=\sum_{w=0}^nx_w\binom{n}{w}3^{n-w}$. We also note that $X$ is clearly Hermitian when $x_w\in\mathbb{R}$. Going forward, we denote the weights by the vector $\textbf{x}$.

Lastly, we consider the positive partial transpose (PPT) constraint. It is straightforward to verify that both twirls map PPT states to PPT states. We find the corresponding constraint by considering the partial transpose of $\Pi_\pm$.
\begin{align}
    \Gamma^B(\Pi_+^{(k)}) & = \frac{1}{2}(\mathbb{1}+\Gamma^B(\mathbb{F}))^{(k)}\\
    & = \frac{1}{2}(\mathbb{1}+\Phi^+)^{(k)}\\
    & = \frac{1}{4}(\Phi^\perp+3\Phi^+)^{(k)},
\end{align}
where $\Phi^+ = \sum_{i,j=0}^1\ketbra{ii}{jj}$ and $\Phi^\perp = 2\mathbb{1}-\Phi^+$. Via similar steps, $\Gamma^B(\Pi_-^{(k)}) = \frac{1}{4}(\Phi^\perp - \Phi^+)^{(k)}$. Thus, $\Gamma^B(X)$ can be written as a summation of terms in $\text{span}\{\Phi^\perp, \Phi^+\}^{\otimes n}$. We will show that this can be represented as a linear transformation $K\textbf{x}$, where $K \in \mathbb{R}^{(n+1)\times \lceil \frac{n+1}{2} \rceil}$. 

Recall that $\tau_w = \text{Perm}\{\Pi_-^{\otimes w}\otimes \Pi_+^{\otimes n-w}\}$, where by "Perm" we indicate a linear combination of all permuted forms. To find a constraint encoding $X^\Gamma \geq 0$ we must find $\tau_w^\Gamma$. Note that
\begin{align}
    \tau_w^{\Gamma}=\frac{1}{4^n}\text{Perm}[(\Phi^\perp-\Phi^+)^{\otimes w}\otimes (\Phi^\perp+3\Phi^+)^{\otimes(n-w)}].\label{Eq:symmetrize-partial-transpose}
\end{align}
It is clear that $(\Phi^\perp-\Phi^+)^{\otimes w}\otimes (\Phi^\perp+3\Phi^+)^{\otimes(n-w)}$ contains every $n$-fold tensor product of $\Phi^\perp$ and $\Phi^+$. Each term in $\text{Perm}[(\Phi^\perp-\Phi^+)^{\otimes w}\otimes (\Phi^\perp+3\Phi^+)^{\otimes(n-w)}]$ will contribute a different number of $\Phi^+$ from $(\Phi^\perp-\Phi^+)^{\otimes w}$ and a remaining number from $(\Phi^\perp+3\Phi^+)^{\otimes(n-w)}$.  Specifically, for each $k=0,\cdots, w$, there will be $\binom{l}{k}\binom{n-l}{w-k}$ terms in $\text{Perm}[(\Phi^\perp-\Phi^+)^{\otimes i}\otimes (\Phi^\perp+3\Phi^+)^{\otimes(n-w)}]$ with exactly $k$ out of $l$ products of $\Phi^+$ coming from $(\Phi^\perp-\Phi^+)^{\otimes w}$ and $l-k$ products of $\Phi^+$ coming from $(\Phi^\perp+3\Phi^+)^{\otimes(n-w)}]$.  Hence, we can write
\begin{align}
    X^\Gamma&=\frac{1}{4^n}\sum_{w \% 2 = 0} x_w\sum_{l=0}^n \sum_{k=0}^w\binom{l}{k}\binom{n-l}{w-k}(-1)^k 3^{l-k}\text{Perm}[(\Phi^+)^{\otimes l}\otimes(\Phi^\perp)^{\otimes n-l}]\notag\\
    &=\frac{1}{4^n}\sum_{l=0}^n\sum_{w \% 2 =0} x_w\left(\sum_{k=0}^w\binom{l}{k}\binom{n-l}{w-k}(-1)^k 3^{l-k}\right)\text{Perm}[(\Phi^+)^{\otimes l}\otimes(\Phi^\perp)^{\otimes n-l}]\notag\\
    &=\frac{1}{4^n}\sum_{l=0}^n\sum_{w \% 2 = 0}x_w3^{l-w} K_{w}(l)\text{Perm}[(\Phi^+)^{\otimes l}\otimes(\Phi^\perp)^{\otimes n-l}],
\end{align}
where
\[K_{w}(l)=\sum_{k=0}^w\binom{l}{k}\binom{n-l}{w-k}(-1)^k 3^{w-k}\]
are the quaternary Krawtchouk polynomials \cite{Roman1992}. Thus, we have $X^\Gamma\geq 0$ iff 
\begin{equation}
    \sum_{w \% 2 = 0} x_w\frac{1}{3^w}K_w(l)\geq 0\,, \qquad\forall l=0,\cdots,n.
\end{equation}
A well-known relation of the Krawtchouk polynomials is
\begin{equation}
    3^l\binom{n}{l}K_w(l)=3^w\binom{n}{w}K_l(w).
\end{equation}
Multiplying both sides of each PPT inequality by $3^l\binom{n}{l}$ allows us to equivalently conclude that $X^\Gamma\geq 0$ iff
\begin{equation}
    \sum_{w \% 2 = 0} x_w\binom{n}{w} K_l(w)\geq 0\,,\qquad\forall l=0,\cdots,n.
\end{equation}
 By defining $y_w := x_w\binom{n}{w}$, the SDP of Eq.~\ref{eq:SDP} can be recast as a LP:
\begin{align}\label{eq:SM_LP}
    \min & \ 3^ny_0\notag\\
    \text{subject to}\;\;& (i)\ \textbf{y} \geq 0,\notag\\
    &(ii)\ Ky \geq 0\,,\notag\\
    &(iii)\ \sum_{w \% 2 = 0} y_w3^{n-w}=1.
\end{align}
Here $K$ is a $(n+1)\times \lceil \frac{n}{2} \rceil$ matrix with elements $K_{i,j} = K_i(2j)$. Note that the equality constraint is simply the last row of $K$. This will become important later. From the constraints, the LP is lower bounded by 0. Further, $\textbf{y} = (\frac{1}{3^n},0,0,\ldots,0)$ is in the feasible region. It is easy to see that this satisfies the positivity and equality constraints. To see that the PPT constraint is satisfied, notice that the first column of $K$ corresponds to $K_i(0)$, which always takes a positive value. Thus, Eq.~\eqref{eq:SM_LP} is bounded and feasible and we expect a solver to find the solution (up to numerical accuracy). Using the optimal value $\textbf{y}^*$, we can upper bound the maximal CE as $\Cmax{n} \leq 1-3^{n}y_0^*$. 

\section{Proof of Proposition 2}
\begin{proposition}\label{prop:SM_haar}
    The Haar average of CE for $s=\mathcal{S}$ is
    \begin{align}\label{eq:SM_haar_av}
        \langle \mathcal{C} \rangle_{\text{Haar}} = 1- 2\frac{3^n}{4^n+2^n},
    \end{align}
    and the Haar variance is
    \begin{align}\label{eq:SM_haar_var}
         \text{var}(\mathcal{C})_{\text{Haar}} = \frac{3^n\cdot4^{1 - n} (2^n - 2\cdot3^n + 4^n)}{(1 + 2^n)^2(6 + 5\cdot2^n + 4^n)}
    \end{align}
\end{proposition}
\subsubsection{Haar Average}
Consider two copies of an arbitrary pure state $X=\op{\psi}^{\otimes 2}$. Using that $\CE{\psi} = 1-\Tr[M\op{\psi}^{\otimes 2}]$, the Haar average can be computed as
\begin{align}
    \langle \mathcal{C} \rangle_{\text{Haar}} = & \int \Tr[1-M(U\otimes U)X(U^\dagger\otimes U^\dagger)] dU \\
    = & \ 1 - \Tr[M\mathcal{T}(X)]\\
    = & \ 1 - \Tr[M\frac{1}{4^n+2^n}(\mathbb{1}+\mathbb{F})]\\
    = & \ 1- 2\frac{3^n}{4^n+2^n},
\end{align}
where the third equality comes from $X$ being a swap symmetric bipartite state.
\subsubsection{Haar Variance}
To compute the variance, we first find a closed form for the second moment.
\begin{align}
    \langle \mathcal{C}^2 \rangle_{\text{Haar}} = & \int (1-\Tr[M(U\otimes U)X(U^\dag\otimes U^\dag)])^2 dU \notag\\
    = & \ \langle \mathcal{C} \rangle_{\text{Haar}} -\frac{3^n}{4^n+2^n} + \int \Tr[(M\otimes M)(U\otimes U\otimes U\otimes U)X\otimes X (U^\dag\otimes U^\dag\otimes U^\dag\otimes U^\dag)]dU\notag\\
    = & \ \langle \mathcal{C} \rangle_{\text{Haar}} -\frac{3^n}{4^n+2^n} + \Tr[M'\mathcal{T}_4(X')],
\end{align}
where $M':=M\otimes M = \bigotimes_k \Pi_+^{(A_kB_k)}\otimes \Pi_+^{(C_kD_k)}$ and $X':=X\otimes X = \op{\psi}^{\otimes 4}$, where $\ket{\psi}$ is again an arbitrary pure state on $n$ qubits. Here by $\mathcal{T}_4$ we denote a 4-partite twirling operation as seen in the second line. From Schur-Weyl duality $\int U^{\otimes k}X (U^\dagger)^{\otimes k}dU$ is a projection onto the space spanned by representations of the symmetry group $S_k$, i.e. $\int U^{\otimes k}X (U^\dagger)^{\otimes k}dU = \sum_{\sigma \in S_k} \mu_\sigma V_\sigma$ \cite{PhysRevA.63.042111}. Here $V_{S_k}$ is again the natural representation of the symmetry group. Thus, $\mathcal{T}_4(X) = \sum_{\sigma \in S_4}\mu_\sigma V_\sigma$. We now show that these coefficients are all equal to a constant $\alpha$. 
\begin{align}
    V_\phi \mathcal{T}_4(X) & = \mathcal{T}_4(V_\phi X)\\
    & = \mathcal{T}_4(X).
\end{align}
The first line follows from $[V_\phi, U^{\otimes 4}]=0$ and the second from $X$ being invariant under all left/right permutations. Thus, $\sum_{\sigma \in S_4}\mu_\sigma V_\sigma = \sum_{\sigma \in S_4}\mu_\sigma V_{\phi\cdot\sigma}$. As this holds for any $\phi\in S_4$, all coefficients are equal. To maintain normalization, $\alpha = (\sum_{\sigma\in S_4}\Tr[V_\sigma])^{-1}$. For $n$-qubit states, this summation works out to be $2^{4n}+6\cdot2^{3n}+11\cdot2^{2n}+6\cdot2^{n}$, as can be seen by computing the number of different types of permutations and their characters (see Tbl.~\ref{tab:perm}). Thus, we find that the second moment takes the form 
\begin{table}[t]
    \centering
    \begin{tabular}{c|c|c}\label{tab:perm}
        Permutation Type & Number & Character\\
        \hline
         (i)(j)(k)(l) &  1 & $d^4$ \\
         (i)(j)(kl) & 6 & $d^3$\\
         (i)(jkl) & 8 & $d^2$\\
         (ij)(kl) & 3 & $d^2$\\
         (ijkl) & 6 & $d$
    \end{tabular}
    \caption{\textbf{Characters of the natural representation of permutation operators in $S_4$.} Here $d$ denotes the local dimension of the subsystems.}
\end{table}
\begin{align}
    \langle \mathcal{C}^2 \rangle_{\text{Haar}} & =  \langle \mathcal{C} \rangle_{\text{Haar}} -\frac{3^n}{4^n+2^n} + \alpha\sum_{\sigma\in \mathcal{S}_4}\Tr[M'V_\sigma].
\end{align}
We now expand the remaining summation:
\begin{align}
    \Tr[M'V_\sigma] & =  \bigotimes_{i=1}^n \Tr[(\Pi_+^{(A_iB_i)}\otimes \Pi_+^{(C_iD_i)})V_\sigma^{(i)}]\\
    & =  \frac{1}{4^n}\Tr[(\mathbb{1}^{(A_iB_iC_iD_i)}+\mathbb{F}^{(A_iB_i)}\otimes\mathbb{1}^{(C_iD_i)}+\mathbb{1}^{(A_iB_i)}\otimes\mathbb{F}^{(C_iD_i)}+\mathbb{F}^{(A_iB_i)}\otimes\mathbb{F}^{(C_iD_i)})V_\sigma^{(i)}]^n,
\end{align}
where $V_\sigma^{(i)}$ is the permutation operator acting upon parties $A_i\otimes B_i \otimes C_i \otimes D_i$. The terms in the parentheses correspond to the subgroup $N\leq S_4$ generated by
\begin{align}
    N = \langle (), (3)(4)(12), (1)(2)(34) \rangle.
\end{align}
We note here that the conjugacy classes of $S_4$ with respect to $N$ will lead to the same value of the trace above. In a slight abuse of notation we denote $\Tr[\sum_{n\in N}V_nV_\sigma]$ as $\Tr[NV_\sigma]$. For $\sigma, \nu \in S_4$ such that $\sigma = n\nu n^{-1}$, where $n\in N$:
\begin{align}
    \Tr[NV_\sigma] & = \Tr[NV_nV_\nu V_n^\dagger]\\
    & = \Tr[V_n^\dagger NV_nV_\nu]\\
    & = \Tr[NV_\nu].
\end{align}
Thus, the conjugacy classes map to the same value and we need only find them and the number of elements in each class. It is clear that $N$ is one such class, with four elements, and $\Tr[N]=36$ (for qubits). Another conjugacy class is $N(ij)$, where $(ij)\notin N$. For these, it can be verified that $\Tr[N(ij)] = 18$. There are four permutations $(ij)\notin N$ and thus 16 elements in this conjugacy class. Lastly, there is the conjugacy class of $(1234)$, which has four elements and $\Tr[N(1234)] = 12$. Thus, $\sum_{\sigma\in S_4}\Tr[M'V_\sigma] = \frac{1}{4^n}(4\cdot36^n+16\cdot18^n+4\cdot12^n)$, which can be readily verified by computing $\Tr[M'V_\sigma]$ for all 24 permutations in $S_4$. Combining the expressions derived thus far, we arrive at
\begin{align}
    \text{var}(\mathcal{C}) = \frac{3^n\cdot4^{1 - n} (2^n - 2\cdot3^n + 4^n)}{(1 + 2^n)^2(6 + 5\cdot2^n + 4^n)}.
\end{align}

\section{Proof of Proposition 3}
\begin{proposition}
    $\Cmax{n}$ goes to 1 like $\Theta ((\frac{3}{4})^n)$.
\end{proposition}
\begin{proof}
    From Eq.~\ref{eq:SM_haar_av} we know that there must exist a state such that $1-\CE{\psi} \in  O((\frac{3}{4})^n)$. We prove below that $\Cmax{n} \leq 1 -(\frac{3}{4})^n$, which yields that $1-\Cmax{n} \in \Omega((\frac{3}{4})^n)$. Thus, the claim is proven.
\end{proof}

\begin{proposition*}
    We note that the following upper bound was given in \cite{cullen2022calculating} as this manuscript was being compiled. We give an alternate proof below. For any system size $n$, the maximum CE can be upper bounded with
    \begin{align}
        \Cmax{n} \leq 1-(\frac{3}{4})^n.
    \end{align}
    \label{prop:asym_bound}
\end{proposition*}
\begin{proof}
    Recall that $\CE{\psi} = 1-\Tr[M\dya{\psi}^{\otimes 2}]$. Now note that $\Pi_+^{(k)}$ is a separable operator, and therefore admits a decomposition into product projectors. For two qubits one finds that
\begin{align}
    \Pi_+^{(k)}=\sum_{t_k=0}^2\op{\varphi_{t_k}}^{\otimes 2},
\end{align}
where $\ket{\varphi_{t_k}}=\frac{1}{\sqrt{2}}(\ket{0}+e^{i2\pi t_k/3}\ket{1})$.  Hence, the CE of $n$-qubit pure state $\ket{\psi}$ is given by
\begin{align}
\label{Eq:CE-quartic}
    \CE{\psi}=1-\sum_{\mathbf{t}\in\mathbb{Z}_3^n}|\bra{\psi}\phi_{\mathbf{t}}\rangle|^4,
\end{align}
where $\ket{\varphi_{\mathbf{t}}}=\ket{\varphi_{t_1}}\ket{\varphi_{t_2}}\cdots\ket{\varphi_{t_n}}$
for $\mathbf{t}=(t_1,t_2,\cdots,t_n)$ with $t_k\in\mathbb{Z}_3=\{0,1,2\}$.


Using the general inequality $\sum_{i=1}^nx_{i}^2\geq \frac{1}{n}\sqrt{\sum_{i=1}^n|x_i|}$, we have
\begin{align}
    \sum_{\mathbf{t}\in\mathbb{Z}_3^{n}}|\bra{\alpha,\beta}\phi_{\mathbf{t}}\rangle|^4&\geq\frac{1}{3^n}\left(\sum_{\mathbf{t}\in\mathbb{Z}_3^{n}}|\bra{\alpha,\beta}\phi_{\mathbf{t}}\rangle|^2\right)^2\notag\\
    &=\frac{1}{3^n}\left(\frac{3}{2}\right)^{2n}=\left(\frac{3}{4}\right)^n,
\end{align}
where we have used the fact that $\sum_{t_k=0}^2\op{\varphi_{t_k}}=\frac{3}{2}\mathbb{1}$. From this the proposition immediately follows.
\end{proof}



\section{Proof of Proposition 4 and Corollary 1}
To prove the proposition we will need the following Lemma:
\begin{lemma}\label{lem:prodprob}
    For a state that can be written as a product state with respect to a partition $P$ of subsystems, i.e. $\ket{\psi} = \otimes_{\alpha \in P}\ket{\psi_\alpha}$ the probability of measuring bitstring $p(\textbf{z})$ in the cSWAP test decomposes as 
    \begin{align}
        p(\textbf{z}) = \prod_{\alpha \in P}p(\textbf{z}_\alpha),
    \end{align}
    where $\textbf{z}_{\alpha}$ is the substring restricted only to parties in $\alpha$ and $p(\textbf{z}_\alpha)$ is the probability of measuring this bitstring when performing the cSWAP test with $\ket{\psi_\alpha}$.
\end{lemma}
\begin{proof}
    We prove the case where $P$ is a bipartition $A | B$. The general claim then follows recursively. By assumption, $\ket{\psi} = \ket{\psi_A}\ket{\psi_{B}}$. For a bistring $z$ we define $\Pi_A$ as $\otimes_{i\in A}\frac{1}{2}(\id+(-1)^{z_i}\mathbb{F})^{(i)}$.
    \begin{align}
        p(\textbf{z}) & = \Tr[\otimes_i\frac{1}{2}(\id+(-1)^{z_i}\mathbb{F})^{(i)}(\ket{\psi_A}\bra{\psi_A}\otimes\ket{\psi_{B}}\bra{\psi_{B}})^{\otimes 2}]\\
        & = \Tr([\Pi_A\otimes\Pi_{B})(\ket{\psi_A}\bra{\psi_A}\otimes\ket{\psi_B}\bra{\psi_B})^{\otimes 2}]\\
        & = \Tr([\Pi_A\ket{\psi_A}\bra{\psi_A}^{\otimes 2}]\Tr([\Pi_B\ket{\psi_B}\bra{\psi_B}^{\otimes 2}]\\
        & = p(\textbf{z}_A)p(\textbf{z}_B).
    \end{align}
\end{proof}
Note that this proof did not use the assumption of pure states anywhere and thus holds for separable mixed states as well.
\begin{proposition}
    Given a measurement of a bitstring $\textbf{z}$ in the cSWAP test, the state cannot be a product with respect to any partition such that the Hamming weight of the substring of $\textbf{z}$ restricted to any set in the partition is odd.
\end{proposition}
\begin{proof}
    In \cite{Beckey2021} it is shown that $p(\textbf{z})=0$ if $w(\textbf{z})$ is odd. Thus, running the SWAP test on a product state will result in $p(\textbf{z}_\alpha) = 0$ if $\textbf{z}_\alpha$ has odd weight. By Lemma~\ref{lem:prodprob} this implies $p(\textbf{z}) = 0$.
\end{proof}
\begin{corollary}
    A measurement of bitstring $\textbf{z}\neq \textbf{0}$ implies that $\ket{\psi}$ cannot be a product state with respect to half of all possible bipartitions. Further, measuring $k$ linearly independent bitstrings implies that a state could only be biseparable with respect to $2^{n-1-k}$ bipartitions.
\end{corollary}
\begin{proof}
    We prove the second claim as it immediately implies the first. A bipartition $A| B$ can be represented as a binary string $a\in Z_2^n$ where $a_i=1$ implies that party $i$ is in the set $A$. The set of bipartitions still possible after measuring bistring $\textbf{z}$ is simply $\{\textbf{a}\in Z_2^n: \langle \textbf{a}, \textbf{z} \rangle =0 \}$. Note that the underlying vector space here is of dimension $n-1$ not $n$ as $\overline{\textbf{a}}$ corresponds to the same bipartition as $\textbf{a}$. The $k$ bitstrings can be arranged into a matrix $\begin{pmatrix} \leftarrow \textbf{z}_1 \rightarrow\\
    \vdots\\
    \leftarrow \textbf{z}_k \rightarrow
    \end{pmatrix}$. Then, the set of possible bipartitions is simply the null space, which via the rank-nullity theorem is of dimension $n-1-k$. The linear span of $n-1-k$ binary strings contains $2^{n-1-k}$ strings, each corresponding to a bipartition.

    
    

\end{proof}

\section{Connections to Other Measures}\label{sect:othermeasures}

\subsection{Proof of Proposition 5 and Corollary 2}
\begin{proposition}\label{prop:SM_cp}
    For a state such that $\text{rk}(\ket{\psi})=R,$
    \begin{align}\label{eq:cp}
        \CE{\psi} \leq 1-\frac{1}{2^n}[\sum_{k=0^n}\binom{n}{k}\max(\frac{1}{R},\frac{1}{2^{\min(k,n-k)}})]
    \end{align}
\end{proposition}
First, we prove the following lemma:
\begin{lemma}
    For a state with CP rank R, no reduced density matrix can have matrix rank greater than R.
\end{lemma}
\begin{proof}
    \begin{align}
        \ket{\psi} & = \sum_{i=1}^R\mu_i\otimes_{j=1}^n\ket{\phi_i^{(j)}}
    \end{align}
    Via linearity, the support of any reduced density matrix $\rho_A$ of $\ket{\psi}$ will be in the span of the corresponding $\otimes_{i\in A}\ket{\phi_i^{(j)}}$'s. As there are R of these terms, the support and thus matrix rank of $\rho_A$ is at most R.
\end{proof}
With this lemma we can readily prove Prop.~\ref{prop:SM_cp}. For a partition A|$\overline{\text{A}}$, $\rho_A$ can have matrix rank up to $2^{\min(|A|,n-|A|)}$. Thus, for partitions of size $2^{\min{|A|,n-|A|}}> R$, $\Tr[\rho_A^2]\geq \frac{1}{R}$. The other terms in the proposition come from assuming all other reduced density matrices are maximally mixed.
\begin{corollary}
    As $n\rightarrow\infty$, a state of CP rank R on n qubits has CE less than $\frac{R-1}{R}$.
\end{corollary}
\begin{proof}
    The corollary follows from assuming $n \gg k$. In this regime, $\binom{n}{k}\approx (\frac{ne}{k})^k(2\pi k)^{-\frac{1}{2}}e^{-\frac{k^2}{2n}(1+O(1))} < (\frac{ne}{k})^k$. Thus, $2^{-n}\binom{n}{k} < 2^{-n}(\frac{ne}{k})^k \rightarrow{} 0$ as $n\rightarrow\infty$. As there are a finite number of terms involving $\binom{n}{h}$ for $h\leq k$, these all are dampened, in the average, to 0 asymptotically. As there are infinitely many terms with corresponding purity $\frac{1}{R}$, Eq.~\eqref{eq:cp} goes to $1-\frac{1}{R}=\frac{R-1}{R}$ as $n\rightarrow \infty$.
\end{proof}
Lastly, we will show that this upper bound is asymptotically tight, as is demonstrated by the family of states $\ket{\psi_R^n} = c\sum_{k=0}^{R-1}R_y(\frac{2\pi k}{R})^{\otimes n}\ket{0^{\otimes n}}$ achieving this value in limit.
    
    We break the proof that this bound is tight into two parts. First, we show that $\CE{\psi_R^n}\rightarrow\frac{R-1}{R}$ as $n\rightarrow\infty$. Second, we show that this implies that the CP rank of these states is actually $R$ (for $n$ large enough).
    
    Note that $\bra{\textbf{0}}R_y^\dagger(\frac{2\pi b}{R})^{\otimes n}R_y(\frac{2\pi a}{R})^{\otimes n}\ket{\textbf{0}} = \cos^n\frac{\pi(a-b)}{R}$. For arbitrary $n$, consider a reduced density matrix on half the system, $\rho_{1,\ldots,\lfloor n/2 \rfloor}$. 
    \begin{align}
        \rho_{1,\ldots,\lfloor n/2 \rfloor} = c^2\sum_{a,b=0}^{R}\cos^n\frac{\pi(a-b)}{R}R_y(\frac{2\pi a}{R})^{\otimes n}\dya{0}^{\otimes n}R_y^\dagger(\frac{2\pi b}{R})^{\otimes n}.
    \end{align}
    The trace of the terms in the summation clearly goes to $\frac{1}{R}$ and thus the normalization constant $c$ also goes to $R^{-\frac{1}{2}}$.
    \begin{align}
    \rho_{1,\ldots,\lfloor n/2 \rfloor}^2 = c^4\sum_{a,b,c,d=0}^R\cos^n\frac{\pi(a-b)}{R}\cos^n\frac{\pi(b-c)}{R}\cos^n\frac{\pi(c-d)}{R}R_y(\frac{2\pi a}{R})^{\otimes n}\dya{0}^{\otimes n}R_y^\dagger(\frac{2\pi d}{R})^{\otimes n}.
    \end{align}
    Thus, as $n\rightarrow\infty$, $\rho_{1,\ldots,\lfloor n/2 \rfloor}^2\rightarrow c^4\sum_{a=0}^R R_y(\frac{2\pi a}{R})^{\otimes n}\dya{0}^{\otimes n}R_y^\dagger(\frac{2\pi a}{R})^{\otimes n}$ and $\Tr[\rho_{1,\ldots,\lfloor n/2 \rfloor}^2]\rightarrow \frac{1}{R}$. By symmetry, thus all reduced density matrices on half the system have purity going to $\frac{1}{R}$.
    
    Now we view $1-\CE{\psi} = 2^{-n}\sum_{\alpha\in Q}\Tr[\rho_\alpha^2]$ as the average of a sequence of $2^n$ terms $a_k$. By ordering this sequence as the $\binom{n}{0}$ and $\binom{n}{n}$ terms, followed by the $\binom{n}{1}$ and $\binom{n}{n-1}$ terms, etc, it is clear that $a_n\rightarrow\frac{1}{R}$. Thus, via the Cesaro mean theorem, the average of the sequence of purities goes to $\frac{1}{R}$ and $\CE{\psi}\rightarrow\frac{R-1}{R}$.
    
    We conclude the proof by showing that these states truly have CP rank $R$. Clearly, they have CP rank at most $R$. Further, for any state $\ket{\phi}$ of rank $R'$ on $n<\infty$ qubits, $\CE{\phi} < \frac{R'-1}{R'}$. As $\CE{\psi_R^n}\rightarrow \frac{R-1}{R}$ and is upper bounded by $\frac{R-1}{R}$ for any finite $n$, for all $R' < R$, there is some $n_{R'}$ such that $\CE{\psi_R^{n_R'}} > \frac{R'-1}{R'}$. Thus, $\text{rnk}(\ket{\psi_R^n}) > R'$ for any $n\geq n_{R'}$ and $R' < R$ via Prop.~\ref{prop:SM_cp}.

\subsection{Bounds on Geometric Measure of Entanglement}
Here we will require two lemmas:

    \begin{lemma}\label{thm:ps_dtr}
        For pure states $\ket{\psi},\ket{\phi}$,
        \begin{align}
            \sqrt{1-|\ip{\phi}{\psi}|^k} \leq \sqrt{k}\sqrt{1-|\ip{\phi}{\psi}|}.
        \end{align}
    \end{lemma}
    \begin{proof}
    We want to find the smallest $\lambda\in\mathbb{R}$ such that $\sqrt{1-x^k}\leq \lambda\sqrt{1-x}$, where $x:=|\ip{\phi}{\psi}|$, for all $\ket{\psi},\ket{\phi}\in\mathcal{H}$.
    \begin{align}
         \lambda^2 & \geq \frac{1-x^k}{1-x}\\
        & = \sum_{i=0}^{k-1}x^i\,.
    \end{align}
    As $x\in[0,1]$, $\sum_{i=0}^{k-1}x^i \leq k$. Thus, if we choose $\lambda:=\sqrt{k}$, then the Lipschitz inequality will hold. Note that this is the smallest possible Lipschitz constant as $\sqrt{1-x^k}\rightarrow \sqrt{k}\sqrt{1-x}$ as $x\rightarrow 1$.
    \end{proof}
    
    In \cite{Schatzki2021} it is shown that $\lvert\CE{\psi}-\CE{\phi}\rvert\leq 2\sqrt{2}D_{tr}(\dya{\psi},\dya{\phi})$, with $D_{tr}$ denoting the trace distance. Here we show that this can be improved. In particular, the following lemma holds. 
    \begin{lemma}
        For pure states $\ket{\psi},\ket{\phi}\in\mathcal{H}_2^{\otimes n}$, the difference in their CE is bounded by
        \begin{align}
            \lvert\CE{\psi}-\CE{\phi}\rvert \leq \sqrt{2}D_{tr}(\dya{\psi},\dya{\phi})
        \end{align}
    \end{lemma}
    \begin{proof}
        \begin{align}
            \lvert\CE{\psi}-\CE{\phi}\rvert & = \lvert \Tr[\Pi_+^{\otimes n}(\dya{\psi}^{\otimes 2}-\dya{\phi}^{\otimes 2})] \rvert\\
            & \leq \frac{1}{2}(\norm{\dya{\psi}^{\otimes 2}-\dya{\phi}^{\otimes 2}}_1+\Tr[\dya{\psi}^{\otimes 2}-\dya{\phi}^{\otimes 2}])\\
            & = \frac{1}{2}(\norm{\dya{\psi}^{\otimes 2}-\dya{\phi}^{\otimes 2}}_1)\\
            & \leq \frac{\sqrt{2}}{2}\norm{\dya{\psi}-\dya{\phi}}_1\\
            & = \sqrt{2}D_{tr}(\dya{\psi},\dya{\phi}).
        \end{align}
    \end{proof}
    The first inequality comes from the the fact that $0 \leq \Pi_+^{\otimes n} \leq \id$ and the definition of trace norm. The second inequality comes from Lemma~\ref{thm:ps_dtr}.
    
    With these Lemmas proven we return to the proposition.

\begin{proposition*}
    For any n-qubit pure state $\ket{\psi}$ its eigenvalue of entanglement, $\Lambda_{\text{max}}$, is upper bounded by $\sqrt{1-\frac{\CE{\psi}^2}{2}}$ and thus $E_g(\ket{\psi}) \geq \frac{\CE{\psi}^2}{2}$.
\end{proposition*}
\begin{proof}
    Take one state to be an arbitrary product state $\ket{\phi}=\bigotimes_i\ket{a_i}$, for which $\CE{\phi} = 0$.
    \begin{align}
        \CE{\psi} & = \lvert\CE{\psi}-\CE{\phi}\rvert\\
        & \leq \sqrt{2}D_{tr}(\dya{\psi},\dya{\phi})\\
        & = \sqrt{2(1-\lvert\ip{\psi}{\phi}\rvert^2)},
    \end{align}
    where the third line comes from $D_{tr}(\dya{\psi},\dya{\phi}) = \sqrt{1-\lvert\ip{\psi}{\phi}\rvert^2}$ for pure states. As this holds for any arbitrary $\ket{\phi}$, rearranging with elementary algebra proves the proposition.
\end{proof}

\section{Expected Number of Bell Pairs}\label{sec:bellpairs}
In the following proofs we will need several known combinatorial identities.
\begin{lemma}\label{lem:comb1}
    \begin{align}
        \sum_{k \text{ even/odd}} \binom{n}{k}= 2^{n-1} \quad \text{and} \quad \sum_{k \text{ even/odd}} k\binom{n}{k} = n2^{n-2}.
    \end{align}
\end{lemma}
\begin{lemma}\label{lem:comb2}
    For $l \geq j \geq 0$ the sum $\sum_{k=0}^l k^j(-1)^k\binom{l}{k}$ evaluates to $0$ if $j\neq l$.
\end{lemma}
We will not prove these here but rather refer the reader to encyclopedias of combinatorial identities \cite{spiegel_lipschutz_liu_2018}.
\begin{proposition}
    Given an n-qubit state $\ket{\psi}$, the expected number of Bell pairs from running the parallelized SWAP test on $\ket{\psi}^{\otimes 2}$ is 
    \begin{align}
        \mathcal{B}(\ket{\psi}) = \frac{1}{2}(n-\sum_i \Tr[\rho_i^2]),
    \end{align}
    with variance
    \begin{align}
        var(\mathcal{B}(\ket{\psi})) = \frac{n}{4}-\frac{\sum_i\Tr[\rho_i^2]}{4}(\sum_j\Tr[\rho_j^2]+2)+\sum_{i,j}\Tr[\rho_{i,j}^2].
    \end{align}
\end{proposition}
\begin{proof}
    The expected number of bell pairs can be written as 
    \begin{align}
        \mathcal{B}(\ket{\psi}) = \sum_{\textbf{z}} w(\textbf{z})*p(\textbf{z}),
    \end{align}
    where $p(\textbf{z})$ is the probability of measuring bitstring $\textbf{z}$ from the cSWAP test. Note that we only need sum over $\textbf{z}$ such that $w(\textbf{z}) \% 2 = 0$ since $p(\textbf{z})=0$ if $\textbf{z}$ has odd Hamming weight \cite{Beckey2021}. As the Kraus operator for a measurement of $0$ ($1$) is $\Pi_+$ ($\Pi_-)$ we can write $p(\textbf{z})$ as 
    \begin{align}
        p(\textbf{z}) = \Tr[(\bigotimes_{i=1}^n \frac{1}{2}(\id+(-1)^{z_i}\mathbb{F}))\ket{\psi}\bra{\psi}^{\otimes 2}].
    \end{align}
    Recall that $\Tr[(\otimes_{i\in \alpha}\mathbb{F}_i)\rho^{\otimes 2}] = \Tr[\rho_\alpha^2]$ for any subset of parties $\alpha$. Thus, $p(\textbf{z}) = \sum_{\alpha\in Q}a(\alpha)\Tr[\rho_\alpha^2]$, where $\{a\}$ is some set of real coefficients. Thus, the expected number of Bell pairs can be written as 
    \begin{align}
        \mathcal{B}(\ket{\psi}) = \sum_{\alpha \in Q} c(\alpha) \Tr[\rho_\alpha^2].
    \end{align}
    Consider two bitstrings $\textbf{z}$ and $\textbf{z'}$ with corresponding sets of coefficients $\{a\}$ and $\{b\}$. If $w(\textbf{z}) = w(\textbf{z'})$ then $\{a\} \cong \{b\}$. This is easy to see by simply applying a permutation that maps each $1$ in $\textbf{z}$ to a $1$ in $\textbf{z'}$. By summing over all $\textbf{z}$, the coefficient for $\Tr[\rho_\alpha^2]$ will be equal to that of $\Tr[\rho_\beta^2]$ if $\alpha$ and $\beta$ are subsystems of the same size. Further, as $\Tr[\rho_\alpha^2] = \Tr[\rho_{\alpha^c}^2]$ we can restrict our attention to subsystems of sizes in the range $[0,\lfloor \frac{n}{2} \rfloor]$. Thus we can write the expected number of Bell pairs as
    \begin{align}
         \mathcal{B}(\ket{\psi}) = 2^{1-n}\sum_{l=0}^{\lfloor \frac{n}{2} \rfloor} c(l) \sum_{\alpha:\lvert \alpha 
         \rvert = l}\Tr[\rho_\alpha^2].
    \end{align}
    Note that if $n$ is even then we  only need to sum over half of the possible subsystems of size $\frac{n}{2}$. This can be absorbed into $c(l)$ and does not affect the following analysis. To compute $c(l)$ we start by noticing that the coefficient for a subsystem will be positive or negative based on the number of $1$'s in $\textbf{z}$ coinciding with parties in $\alpha$. In particular, each $1$ adds a factor of $-1$. For a given system of size $l$ and bitstring $\textbf{z}$, there are $k=0,1,\ldots,\min(l,w(z))$ possible overlaps. Thus, we need to sum over the possible weights and the number of ways that we could construct bitstrings such that $\alpha$ and $\textbf{z}$ overlap in $k$ locations. One can verify that
    \begin{align}
    c(l) = \sum_{w\% 2 = 0}w*\sum_{k=0}^w(-1)^k\binom{l}{k}\binom{n-l}{w-k}.
\end{align}
We now manipulate this expression into a more manageable form.
\begin{align}
     c(l) & = \sum_{w\% 2 = 0}\sum_{k=0}^w w(-1)^k\binom{l}{k}\binom{n-l}{w-k}\\
     & = \sum_{w\% 2 = 0}\sum_{k=0}^w (-1)^k\binom{l}{k}(w-k+k)\binom{n-l}{w-k}\\
     & = \sum_{k=0}^{l}(-1)^k\binom{l}{k}
     \sum_{\substack{w\% 2 = 0 \\ w\geq k}}^{n-l+k} (w-k+k)\binom{n-l}{w-k}.
\end{align}
Fix $i:= w-k$ and continue:
\begin{align}
    c(l) & = \sum_{k=0}^l(-1)^k \binom{l}{k}(\sum_{\substack{i=0\\ i + k \% 2 = 0}}^{n-l}i\binom{n-l}{i}+k\sum_{\substack{i=0\\ i + k \% 2 = 0}}^{n-l}\binom{n-l}{i})\\
     & = \sum_{k=0}^l(-1)^k \binom{l}{k}((n-l)2^{n-l-2}+k*2^{n-l-1})\\
     & = 2^{n-l-1}\sum_{k=0}^l(-1)^k\binom{l}{k}(\frac{n-l}{2}+k)\\
     & = (n-l)2^{n-l-2}\sum_{k=0}^l(-1)^k\binom{l}{k}+2^{n-l-1}\sum_{k=0}^lk(-1)^k\binom{l}{k}\,.
\end{align}
In the fifth equality we are using Lemma~\ref{lem:comb1}. To finish the proof, note that $\sum_{k=0}^l k^j(-1)^k\binom{l}{k}$ evaluates to $0$ if $j\neq l$ via Lemma~\ref{lem:comb2}. For $l=0$ or $1$ it is easy to evaluate this expression to find
\begin{align}
    c(l) = \begin{cases}
        n2^{n-2} & l = 0\\
        -2^{n-2} & l = 1\\
        0 & \text{otherwise}
    \end{cases}
\end{align}
Plugging this into our expression for $\mathcal{B}(\ket{\psi})$ yields the desired identity.

Let us now evaluate the variance. Here we must find an expression for the second moment of the number of Bell pairs. This can also be written as some summation over purities. Following the same argument as above, one can show that
\begin{align}
    \sum_{\textbf{z}}w(\textbf{z})^2p(\textbf{z}) = 2^{1-n}\sum_{l=0}^{\lfloor \frac{n}{2} \rfloor}d(l)\sum_{\alpha:\lvert \alpha \rvert =l}\Tr[\rho_\alpha^2],
\end{align}
where the same combinatorial argument as above yields 
\begin{align}
    d(l) & = \sum_{w\% 2 = 0}w^2*\sum_{k=0}^w(-1)^k\binom{l}{k}\binom{n-l}{w-k}\\
    & = \sum_{k=0}^l(-1)^k \binom{l}{k}(\sum_{\substack{i=0\\ i + k \% 2 = 0}}^{n-l}(i+k)^2\binom{n-l}{i} \\
    & = \sum_{k=0}^l(-1)^k \binom{l}{k}[(n-l)(n-l+1)2^{n-l-3}+k(n-l)2^{n-l-1}+k^22^{n-l}].
\end{align}
 Via Lemma~\ref{lem:comb2}, the first term is nonzero only for $l=0$, the second for $l=1$, and the third for $l=2$. Evaluating, we arrive at
    \begin{align}
        d(l) = \begin{cases}
            n(n+1)2^{n-3} & l = 0\\
        (1-n)2^{n-2} & l = 1\\
        2^{n-1} & l = 2\\
        0 & \text{otherwise}
        \end{cases}
    \end{align}
    Thus the second moment is $\frac{n(n+1)}{4}-\frac{n-1}{2}\sum_i\Tr[\rho_i^2]+\sum_{i,j}\Tr[\rho_{i,j}^2]$. We can now compute the variance:
    \begin{align}
        \text{var}(B)_{\ket{\psi}} = \frac{n}{4}-\frac{\sum_i\Tr[\rho_i^2]}{4}(\sum_j\Tr[\rho_j^2])+\sum_{i,j}\Tr[\rho_{i,j}^2].
    \end{align}
\end{proof}
\begin{proposition}
    The k\textsuperscript{th} moment of the number of Bell pairs generated in the SWAP test, $\sum_\textbf{z} w(\textbf{z})^mp(\textbf{z})$, is a function of the purities of only $1-m$ party marginals and not of that of any larger subsystems.
\end{proposition}
\begin{proof}
    As in the case of $m=1$ this evaluates to some sum over subsystem purities. By the same argument as before, the coefficients will be the same for all subsystems of the same size. 
    \begin{align}
        \sum_\textbf{z} w(\textbf{z})^mp(\textbf{z}) & = \sum_{\alpha \in Q}d(\alpha) \Tr[\rho_\alpha^2]\\
        & = \frac{1}{2^{n-1}}\sum_{l=0}^{\lfloor \frac{n}{2} \rfloor} d(l) \sum_{\alpha:\lvert \alpha 
         \rvert = l}\Tr[\rho_\alpha^2].
    \end{align}
    Using the same argument as $m=1$ it is clear that
    \begin{align}
        d(l) = \sum_{w \% 2 =0} w^m*\sum_{k=0}^w(-1)^k\binom{l}{k}\binom{n-l}{w-k}.
    \end{align}
    Via the same steps as before and the substitution $i:=w-k$ this can be manipulated into the form
    \begin{align}
        d(l) = \sum_{k=0}^l(-1)^k\left(\sum_{\substack{i=0 \\ i+k \% 2 = 0}}^{n-l}\binom{l}{k}(i+k)^m\binom{n-l}{i}\right).
    \end{align}

Using the binomial theorem we can readily expand this to
\begin{align}
    d(l) = \sum_{k=0}^l(-1)^k \sum_{j=0}^m \binom{m}{j}k^j\left(\sum_{\substack{i=0 \\ i+k \% 2 = 0}}^{n-l}i^{m-j}\binom{n-l}{i}\right).
\end{align}
From Lemma~\ref{lem:comb2} we know that each term $ \sum_{k=0}^l(-1)^k \binom{m}{j}k^j$ is non-zero only if $l=j$. As we are summing from $j = 0 $ to $m$, $d(l) =0$ for $l > m$.
\end{proof}

\section{Connections to Coding Theory}\label{sec:coding}
\subsection{States that Maximize Concentratable Entanglement}
\begin{figure}[t]
    \centering
    \includegraphics[width=1.\columnwidth]{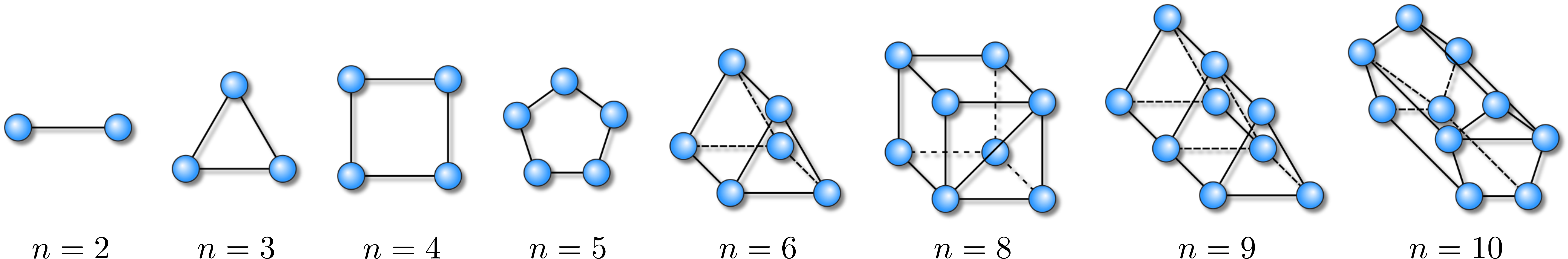}
    \caption{\textbf{Graph states that achieve $\CE{G}=\Cmax{n}$}. Note that for $n=7$ graph states are not maximal. Every state given corresponds to a type I code except for $n=2$ or $n=6$, which are type II. Further, each code is extremal, i.e., the distance is 2 for $n\leq 4$, 3 for $n=5$, and 4 for $n\geq 6$. }
    \label{fig:graphs}
\end{figure}
First, let us recall that one can easily convert between the weight enumerator of a code and the purities of subsystems, as is detailed in \cite{Scott2004}. We begin by noting that the error correcting codes for $n\in\{2,3,4,5,6,7,8,9,11,12\}$ given in Tbl. 1 of \cite{Scott2004} have CE equal to the upper bound. Here we give explicit constructions, shown as graph states in Fig.~\ref{fig:graphs}, corresponding to these codes and a different code for $n=10$ which achieves the CE upper bound. We will quickly recall the definition of graph states here and direct the reader to \cite{Hein2006} for a detailed review. Note that a recent work considered how to calculate subsystem purities for graph states and found the maximal values over graph states for up to 9 qubits. Here we show that the certain graphs states are indeed maximal for their system size over all states.

A graph state is a pure state which can be prepared given a graph $G$ with edge set $E$:
\begin{align}
    \ket{G} = \Pi_{(a,b)\in E}U^{(a,b)}\ket{+}^{\otimes n},
\end{align}
where $U^{(a,b)}$ is the controlled-Z operation. Graph states are also stabilizer states and thus $[[n,0,d]]$ quantum error correcting codes. The $n$ stabilizers are given by
\begin{align}
    S_b = \sigma_x^{(b)}\Pi_{v\in N(b)}\sigma_z^{(v)},
\end{align}
where $N(b)$ is the neighborhood of the node $b$ in the graph $G$. Thus, the stabilizer matrix for a graph state will take the simple form $(\mathbb{1}|\Gamma)$, where $\Gamma$ is the graph adjacency matrix. We note that every state given here corresponds to an extremal code. Lastly, we note that the optimal $n=7$ code in \cite{Scott2004} achieves the maximum CE over all stabilizer states, which is $0.7734375$ and corresponds to the maximum number (32 out of 35) of 3 party reduced density matrices being maximally mixed \cite{Huber2017}. However, this does not achieve $\mathcal{C}^*(7)$.

More generally, by considering extremal codes over GF(4), we find that the bounds from the LP for $n\in\{14,17,18\}$ are tight as well. The length 14 code denoted as $G_{14,2}$ in \cite{Varbanov2007} achieves the bound from the LP. As does the length 17 code ($C_{17,1}$) in \cite{1266813} and the length 18 code ($S18$) in \cite{BACHOC200155}.

\subsection{Bounds}
In Eq.~\eqref{eq:SM_LP} we provided a linear program to upper bound $\Cmax{n}$. As mentioned in the main text, the inequality constraint $Ky \geq 0$ is the quantum MacWilliams identity. Now we also note that the equality constraint is a sum over $K_n(2j)$. This suggests that if $\{y\}$ is some enumerator $\{A\}$ such that all odd enumerators are 0, then the equality constraint is simply the transformation from $\{A\}$ to $B_n$. We will work with this assumption for now and return to the general case later. 

The quantum MacWilliams identity yields
\begin{align}
    B_n = 2^{k-n}\sum_{j}K_n(j)A_j.
\end{align}
If we now assume that $y_i = A_{2i}$ this yields the condition
\begin{align}
    B_n = 2^{k-n}.
\end{align}
As $k \leq n$ this is not a valid statement. However, we can rescale Eq.~\eqref{eq:SM_LP} by some $\alpha > 0$ such that this constraint is valid. That is, instead of optimizing over $\{y\}$ we optimize over $\{A\}$ where $A_i = \alpha y_i$. Note that rescaling does not affect the two inequality constraints and simply multiplies the optimal value by $\alpha$. We choose this factor to be $\alpha := 2^{n-k}B_n$ so that $2^{n-k}B_n = \sum_{w\% 2 =0}3^{n-w}A_w = \alpha = 2^{n-k}B_n$. Thus, given some enumerator $B_n$ we can rescale the value of the LP, $L(n)$, to yield a bound on the code.
\begin{proposition}
   A stabilizer code such that $A_{2i+1}=0$ must satisfy the inequality
    \begin{align}
        B_n \leq \frac{1}{L(n)}\frac{3^n}{2^{n-k}}.
    \end{align}
\end{proposition}
\begin{proof}
    The objective in Eq.~\eqref{eq:SM_LP} is constant for an enumerator, as $A_0$ is always equal to one. Thus, we need to replace it via the equality constraint to obtain $3^ny_0 = 1- \sum_{j > 0}K_n(2j)y_j$. Scaling by $\alpha = 2^{n-k}B_n$ this yields
    \begin{align}
        \alpha - \sum_{j>0}K_n(2j)A_{2j} \geq \alpha L(n) = 2^{n-k}B_nL(n).
    \end{align}
    The left hand side can be rewritten as $ 2^{n-k}B_n + 3^n - \sum_{j=0}^{\lfloor \frac{n}{2} \rfloor}K_n(2j)A_{2j} = 3^n$. Rearranging yields the desired inequality.
\end{proof}

\begin{proposition}
    A stabilizer code such $A_{2i+1}=0$ must satisfy the inequality
    \begin{align}
        \sum_{i=0}^n i\frac{1}{3^i}A_i \leq \frac{n}{4}\frac{2^{n-k}}{3^n}B_n.
    \end{align}
\end{proposition}
\begin{proof}
    In a similar fashion to the LP approach for CE, we consider a relaxation to PPT states and use symmetry (unitary and symmetry group twirls) to arrive at 
    \begin{align}
        \max & \quad \sum_{w\% 2 = 0} w*3^{n-w}*y_w\\
        & \quad y\geq 0,\notag\\
        \text{subject to} & \quad Ky\geq 0,\notag\\
        & \quad \sum_{w\% 2 = 0}3^{n-w}y_w = 1.\notag
    \end{align}
    
    As this is a relaxation of $\max_{\ket{\psi}} \mathcal{B}(\ket{\psi})$, the value $\frac{n}{4}$ is achievable since 1-uniform states exist for any system size $n \geq 2$. We now show that $\frac{n}{4}$ is dual feasible and thus is the value of this linear program. The dual program is
    \begin{align}
    \text{minimize} & \quad \frac{1}{3^n}\nu\\
    \text{subject to}& \quad \nu d \geq b + K^T\lambda\,,\notag\\
    &\quad \lambda \geq 0\notag,
\end{align}
where $d_i = 3^{-2i}$ and $b_i = 2i*3^{n-2i}$. Note that $b,d\in\mathbb{R}^{\lceil (n+1)/2\rceil}$, $K \in \mathbb{R}^{(n+1) \times \lceil (n+1)/2\rceil}$, and $\lambda\in\mathbb{R}^{n+1}$. We claim that $\nu = \frac{n}{4}3^{n}$ and $\lambda = \frac{1}{4}e_{n-1}$ lies in the feasible region (here $e_j$ is the standard basis vector with entries $(e_j)_i = \delta_{i,j}$). For these choices, we see that
\begin{align}
    (K^T\lambda)_i = \frac{1}{4}K_{n-1}(2i;n),\\
    (b+K^T\lambda)_i = 2i*3^{n-2i}+\frac{1}{4}K_{n-1}(2i;n),\\
    (\nu d)_i = \frac{n}{4}3^{n-2i}.
\end{align}
By expanding the Krawtchouk polynomials one can verify that
\begin{align}
    K_{n-1}(2i;n) = 3^{n-2i}(n-8i).
\end{align}
Thus, $\nu d = b+K^T\lambda$ and this point lies in the feasible region. Evaluating, this shows that $\frac{n}{4}$ is dual feasible.

As above, we now introduce a scaling parameter $\alpha := 2^{n-k}B_n$ so that we can take $\{y\}$ to be the enumerators of a code. As we are assuming that $A_{2i+1}=0$, this yields the bound $\sum_ii*3^{n-i}A_i \leq \frac{n}{4}2^{n-k}B_n$. Moving the factor of $3^n$ to the right hand side yields the desired inequality.
\end{proof}

In the propositions above we assumed that only the even weight enumerators $\{A_{2j}\}$ could be non-zero. However, our linear programs yield more general bounds. For convenience we define $C = \sum_{w\% 2 =0}K_n(w)A_w$. This is (up to normalization) the equality constraint arising in our linear programs. As before, we can scale the objective  to obtain a bound on stabilizer codes.

First we consider the LP for bounding $\Cmax{n}$:

\begin{proposition*}
    Any stabilizer code must satisfy the inequality $C \leq \frac{3^n}{L(n)}$, where $C = \sum_{w \% 2=0} K_n(w)A_w$.
\end{proposition*}

Similarly we can write the bound given from the Bell pair LP as:
\begin{proposition*}
    Any stabilizer code must satisfy the inequality $\sum_{w \% 2 = 0}w\frac{1}{3^w}A_{w} \leq \frac{n}{4}\frac{C}{3^n}$, where $C = \sum_{w \% 2=0} K_n(w)A_w$.
\end{proposition*}

\newpage
\section{Results from Linear Programming and Hierarchies}\label{data}

\subsection{Results from Linear Programming}
Below we list results from solving the LP given in Eq.~\eqref{eq:SM_LP} for up to 31 qubits. Results are reported for up to 15 decimal places. As we only know of states that maximize CE for $n\leq 12$, these are generally upper bounds for $\Cmax{n}$.

\begin{center}
\begin{tabular}{c|c|c|c|c|c}
     Number of Qubits & $\Cmax{n}$ & Number of Qubits & $\Cmax{n}$ & Number of Qubits \\
     \hline
     2 & 0.25 & 12 & 0.94287109375 & 22 & 0.99652671813964844964\\
    3 & 0.375 & 13 & 0.956512451171875 & 23 & 0.99739503860473658096\\
    4 & 0.5 & 14 & 0.96685791015625 & 24 & 0.99804663658142089974\\
    5 & 0.625 & 15 & 0.974822998046875 & 25 & 0.99852997395727370144\\
    6 & 0.71875 & 16 & 0.98095703125 & 26 & 0.99889324605464935585\\
    7 & 0.779296875 & 17 & 0.9857177734375 & 27 & 0.9991672709584236147\\
    8 & 0.828125 & 18 & 0.989288330078125 & 28 & 0.9993754532188177101\\
    9 & 0.8671875 & 19 & 0.9919147491455078333 & 29 & 0.99953059107065201\\
    10 & 0.8984375 & 20 & 0.99388885498046875694 & 30 & 0.99965104753916453405\\
    11 & 0.923828125 & 21 & 0.99538803100585938107 & 31 & 0.9997371863573789609
\end{tabular}
\end{center}

\subsection{Hierarchies on 3-12 Qubits}\label{hiers}
Here we give hierarchies for up to 12 qubits based on CE. The values come from iteratively applying Prop.~\ref{prop:SM_bisep} and using the upper bounds on $\Cmax{n}$ found via linear programming. Below, $a\otimes b\otimes c\otimes \ldots$ indicates the state can be written as a tensor product of a pure state on $a$ qubits with a state on $b$ of qubits, with a state on $c$ qubits, etc. Note that every value is tight except for those with a partition of size 7.

\begin{center}
    \begin{tabular}{c|c}
    Structure (in qubits) & Max CE\\
    \hline
     3   &  0.375\\
     $2\otimes 1$   & 0.25\\
     $1\otimes 1\otimes 1$ & 0
    \end{tabular}
\end{center}
\begin{center}
    \begin{tabular}{c|c}
        Structure & Max CE\\
        \hline
        4 & 0.5\\
        $3\otimes 1$ & 0.375\\
        $2\otimes 2$ & 0.4375\\
        $2\otimes 1 \otimes 1$ & 0.25 \\
        $1\otimes 1 \otimes 1\otimes 1$ & 0
    \end{tabular}
\end{center}
\begin{center}
    \begin{tabular}{c|c}
        Structure & Max CE\\
        \hline
        5 & 0.6245\\
        $3\otimes 2$ & 0.53125\\
        $4\otimes 1$ & 0.5\\
        $2\otimes 2 \otimes 1$ & 0.4375\\
        $3\otimes 1 \otimes 1$ & 0.375 \\
        $2\otimes 1\otimes 1\otimes 1\otimes 1$ & 0.25\\
        $1\otimes 1\otimes 1\otimes 1\otimes 1$ & 0
    \end{tabular}
\end{center}
\begin{center}
    \begin{tabular}{c|c}
        Structure & Max CE\\
        \hline
        6 & 0.71875\\
        $5\otimes 1$ & 0.625\\
        $4\otimes 2$ & 0.625\\
        $3\otimes 3$ & 0.609375\\
        $2\otimes 2\otimes 2$ & 0.578125\\
        $3\otimes 2 \otimes 1$ & 0.53125\\
        $4\otimes 1\otimes 1$& 0.5\\
        $2\otimes 2\otimes 1\otimes 1$ & 0.4375\\
        $3\otimes 1\otimes 1\otimes 1\otimes 1$ & 0.375\\
        $2\otimes 1\otimes 1\otimes 1\otimes 1\otimes 1$ & 0.25\\
        $1\otimes 1\otimes 1\otimes 1\otimes 1\otimes 1\otimes 1$ & 0\\
    \end{tabular}
\end{center}
\begin{center}
    \begin{tabular}{c|c}
        Structure & Max CE\\
        \hline
        7 & 0.779296875\\
        $6\otimes 1$ & 0.71825\\
        $5\otimes 2$ & 0.71825\\
        $4\otimes 3$ & 0.7\\
        $3\otimes 2\otimes 2\otimes 1$ & 0.6484375\\
        $5\otimes 1\otimes 1$& 0.625\\
        $4\otimes 2 \otimes 1$ & 0.625\\
        $3\otimes 3\otimes 1$ & 0.609375\\
        $2\otimes 2\otimes 2\otimes 1$ & 0.578125\\
        $3\otimes 2\otimes 1\otimes 1\otimes 1\otimes 1$ & 0.53125\\
        $4\otimes 1\otimes 1\otimes 1\otimes 1$ & 0.5\\
        $2\otimes 2\otimes 1\otimes 1\otimes 1$ & 0.4375\\
        $3\otimes 1\otimes 1\otimes 1\otimes 1$ & 0.375\\
        $2\otimes 1\otimes 1\otimes 1\otimes 1\otimes 1$ & 0.25\\
        $1\otimes 1\otimes 1\otimes 1\otimes 1\otimes 1\otimes 1$ & 0
        \end{tabular}
\end{center}
\begin{center}
    \begin{tabular}{c|c}
    Structure & Max CE\\
    \hline
    $8$ & 0.828125\\
$6\otimes 2$ & 0.7890625\\
$7\otimes 1$ & 0.779296875\\
$5\otimes 3$ & 0.765625\\
$4\otimes 4$ & 0.75\\
$6\otimes 1\otimes 1$ & 0.71875\\
$5\otimes 2\otimes 1$ & 0.71875\\
$4\otimes 2\otimes 2$ & 0.71875\\
$3\otimes 3\otimes 2$ & 0.70703125\\
$4\otimes 3\otimes 1$ & 0.6875\\
$2\otimes 2\otimes 2\otimes 2$ & 0.68359375\\
$3\otimes 2\otimes 2\otimes 1$ & 0.6484375\\
$5\otimes 1\otimes 1\otimes 1$ & 0.625\\
$4\otimes 2\otimes 1\otimes 1$ & 0.625\\
$3\otimes 3\otimes 1\otimes 1$ & 0.609375\\
$2\otimes 2\otimes 2\otimes 1\otimes 1$ & 0.578125\\
$3\otimes 2\otimes 1\otimes 1\otimes 1$ & 0.53125\\
$4\otimes 1\otimes 1\otimes 1\otimes 1$ & 0.5\\
$2\otimes 2\otimes 1\otimes 1\otimes 1\otimes 1$ & 0.4375\\
$3\otimes 1\otimes 1\otimes 1\otimes 1\otimes 1$ & 0.375\\
$2\otimes 1\otimes 1\otimes 1\otimes 1\otimes 1\otimes 1$ & 0.25\\
$1\otimes 1\otimes 1\otimes 1\otimes 1\otimes 1\otimes 1\otimes 1$ & 0
    \end{tabular}
\end{center}

\begin{center}
    \begin{tabular}{c|c}
    Structure & Max CE\\
    \hline
    $9$ & 0.8671875\\
$7\otimes 2$ & 0.83447265625\\
$8\otimes 1$ & 0.828125\\
$6\otimes 3$ & 0.82421875\\
$5\otimes 4$ & 0.8125\\
$6\otimes 2\otimes 1$ & 0.7890625\\
$5\otimes 2\otimes 2$ & 0.7890625\\
$7\otimes 1\otimes 1$ & 0.779296875\\
$5\otimes 3\otimes 1$ & 0.765625\\
$4\otimes 3\otimes 2$ & 0.765625\\
$3\otimes 3\otimes 3$ & 0.755859375\\
$4\otimes 4\otimes 1$ & 0.75\\
$3\otimes 2\otimes 2\otimes 2$ & 0.736328125\\
$6\otimes 1\otimes 1\otimes 1$ & 0.71875\\
$5\otimes 2\otimes 1\otimes 1$ & 0.71875\\
$4\otimes 2\otimes 2\otimes 1$ & 0.71875\\
$3\otimes 3\otimes 2\otimes 1$ & 0.70703125\\
$4\otimes 3\otimes 1\otimes 1$ & 0.6875\\
$2\otimes 2\otimes 2\otimes 2\otimes 1$ & 0.68359375\\
$3\otimes 2\otimes 2\otimes 1\otimes 1$ & 0.6484375\\
$5\otimes 1\otimes 1\otimes 1\otimes 1$ & 0.625\\
$4\otimes 2\otimes 1\otimes 1\otimes 1$ & 0.625\\
$3\otimes 3\otimes 1\otimes 1\otimes 1$ & 0.609375\\
$2\otimes 2\otimes 2\otimes 1\otimes 1\otimes 1$ & 0.578125\\
$3\otimes 2\otimes 1\otimes 1\otimes 1\otimes 1$ & 0.53125\\
$4\otimes 1\otimes 1\otimes 1\otimes 1\otimes 1$ & 0.5\\
$2\otimes 2\otimes 1\otimes 1\otimes 1\otimes 1\otimes 1$ & 0.4375\\
$3\otimes 1\otimes 1\otimes 1\otimes 1\otimes 1\otimes 1$ & 0.375\\
$2\otimes 1\otimes 1\otimes 1\otimes 1\otimes 1\otimes 1\otimes 1$ & 0.25\\
$1\otimes 1\otimes 1\otimes 1\otimes 1\otimes 1\otimes 1\otimes 1\otimes 1$ & 0
    \end{tabular}
\end{center}

\begin{center}
   \begin{tabular}{c|c}
    Structure & Max CE\\
    \hline
$10$ & 0.8984375\\
$8\otimes 2$ & 0.87109375\\
$9\otimes 1$ & 0.8671875\\
$7\otimes 3$ & 0.862060546875\\
$6\otimes 4$ & 0.859375\\
$5\otimes 5$ & 0.859375\\
$6\otimes 2\otimes 2$ & 0.841796875\\
$7\otimes 2\otimes 1$ & 0.83447265625\\
$8\otimes 1\otimes 1$ & 0.828125\\
$6\otimes 3\otimes 1$ & 0.82421875\\
$5\otimes 3\otimes 2$ & 0.82421875\\
$5\otimes 4\otimes 1$ & 0.8125\\
$4\otimes 4\otimes 2$ & 0.8125\\
$4\otimes 3\otimes 3$ & 0.8046875\\
$6\otimes 2\otimes 1\otimes 1$ & 0.7890625\\
$5\otimes 2\otimes 2\otimes 1$ & 0.7890625\\
$4\otimes 2\otimes 2\otimes 2$ & 0.7890625\\
$3\otimes 3\otimes 2\otimes 2$ & 0.7802734375\\
$7\otimes 1\otimes 1\otimes 1$ & 0.779296875\\
$5\otimes 3\otimes 1\otimes 1$ & 0.765625\\
$4\otimes 3\otimes 2\otimes 1$ & 0.765625\\
$2\otimes 2\otimes 2\otimes 2\otimes 2$ & 0.7626953125\\
$3\otimes 3\otimes 3\otimes 1$ & 0.755859375\\
$4\otimes 4\otimes 1\otimes 1$ & 0.75\\
$3\otimes 2\otimes 2\otimes 2\otimes 1$ & 0.736328125\\
$6\otimes 1\otimes 1\otimes 1\otimes 1$ & 0.71875\\
$5\otimes 2\otimes 1\otimes 1\otimes 1$ & 0.71875\\
$4\otimes 2\otimes 2\otimes 1\otimes 1$ & 0.71875\\
$3\otimes 3\otimes 2\otimes 1\otimes 1$ & 0.70703125\\
$4\otimes 3\otimes 1\otimes 1\otimes 1$ & 0.6875\\
$2\otimes 2\otimes 2\otimes 2\otimes 1\otimes 1$ & 0.68359375\\
$3\otimes 2\otimes 2\otimes 1\otimes 1\otimes 1$ & 0.6484375\\
$5\otimes 1\otimes 1\otimes 1\otimes 1\otimes 1$ & 0.625\\
$4\otimes 2\otimes 1\otimes 1\otimes 1\otimes 1$ & 0.625\\
$3\otimes 3\otimes 1\otimes 1\otimes 1\otimes 1$ & 0.609375\\
$2\otimes 2\otimes 2\otimes 1\otimes 1\otimes 1\otimes 1$ & 0.578125\\
$3\otimes 2\otimes 1\otimes 1\otimes 1\otimes 1\otimes 1$ & 0.53125\\
$4\otimes 1\otimes 1\otimes 1\otimes 1\otimes 1\otimes 1$ & 0.5\\
$2\otimes 2\otimes 1\otimes 1\otimes 1\otimes 1\otimes 1\otimes 1$ & 0.4375\\
$3\otimes 1\otimes 1\otimes 1\otimes 1\otimes 1\otimes 1\otimes 1$ & 0.375\\
$2\otimes 1\otimes 1\otimes 1\otimes 1\otimes 1\otimes 1\otimes 1\otimes 1$ & 0.25\\
$1\otimes 1\otimes 1\otimes 1\otimes 1\otimes 1\otimes 1\otimes 1\otimes 1\otimes 1$ & 0
    \end{tabular}
\end{center}

\begin{small}
\begin{center}
    \begin{tabular}{c|c}
    Structure & Max CE \\
    \hline
     $11$ & 0.92382812\\
$9\otimes 2$ & 0.900390625\\
$10\otimes 1$ & 0.8984375\\
$6\otimes 5$ & 0.89453125\\
$8\otimes 3$ & 0.892578125\\
$7\otimes 4$ & 0.8896484375\\
$7\otimes 2\otimes 2$ & 0.8758544921875\\
$8\otimes 2\otimes 1$ & 0.87109375\\
$6\otimes 3\otimes 2$ & 0.8681640625\\
$9\otimes 1\otimes 1$ & 0.8671875\\
$7\otimes 3\otimes 1$ & 0.862060546875\\
$6\otimes 4\otimes 1$ & 0.859375\\
$5\otimes 5\otimes 1$ & 0.859375\\
$5\otimes 4\otimes 2$ & 0.859375\\
$5\otimes 3\otimes 3$ & 0.853515625\\
$4\otimes 4\otimes 3$ & 0.84375\\
$6\otimes 2\otimes 2\otimes 1$ & 0.841796875\\
$5\otimes 2\otimes 2\otimes 2$ & 0.841796875\\
$7\otimes 2\otimes 1\otimes 1$ & 0.83447265625\\
$8\otimes 1\otimes 1\otimes 1$ & 0.828125\\
$6\otimes 3\otimes 1\otimes 1$ & 0.82421875\\
$5\otimes 3\otimes 2\otimes 1$ & 0.82421875\\
$4\otimes 3\otimes 2\otimes 2$ & 0.82421875\\
$3\otimes 3\otimes 3\otimes 2$ & 0.81689453125\\
$5\otimes 4\otimes 1\otimes 1$ & 0.8125\\
$4\otimes 4\otimes 2\otimes 1$ & 0.8125\\
$4\otimes 3\otimes 3\otimes 1$ & 0.8046875\\
$3\otimes 2\otimes 2\otimes 2\otimes 2$ & 0.80224609375\\
$6\otimes 2\otimes 1\otimes 1\otimes 1$ & 0.7890625\\
$5\otimes 2\otimes 2\otimes 1\otimes 1$ & 0.7890625\\
$4\otimes 2\otimes 2\otimes 2\otimes 1$ & 0.7890625\\
$3\otimes 3\otimes 2\otimes 2\otimes 1$ & 0.7802734375\\
$7\otimes 1\otimes 1\otimes 1\otimes 1$ & 0.779296875\\
$5\otimes 3\otimes 1\otimes 1\otimes 1$ & 0.765625\\
$4\otimes 3\otimes 2\otimes 1\otimes 1$ & 0.765625\\
$2\otimes 2\otimes 2\otimes 2\otimes 2\otimes 1$ & 0.7626953125\\
$3\otimes 3\otimes 3\otimes 1\otimes 1$ & 0.755859375\\
$4\otimes 4\otimes 1\otimes 1\otimes 1$ & 0.75\\
$3\otimes 2\otimes 2\otimes 2\otimes 1\otimes 1$ & 0.736328125\\
$6\otimes 1\otimes 1\otimes 1\otimes 1\otimes 1$ & 0.71875\\
$5\otimes 2\otimes 1\otimes 1\otimes 1\otimes 1$ & 0.71875\\
$4\otimes 2\otimes 2\otimes 1\otimes 1\otimes 1$ & 0.71875\\
$3\otimes 3\otimes 2\otimes 1\otimes 1\otimes 1$ & 0.70703125\\
$4\otimes 3\otimes 1\otimes 1\otimes 1\otimes 1$ & 0.6875\\
$2\otimes 2\otimes 2\otimes 2\otimes 1\otimes 1\otimes 1$ & 0.68359375\\
$3\otimes 2\otimes 2\otimes 1\otimes 1\otimes 1\otimes 1$ & 0.6484375\\
$5\otimes 1\otimes 1\otimes 1\otimes 1\otimes 1\otimes 1$ & 0.625\\
$4\otimes 2\otimes 1\otimes 1\otimes 1\otimes 1\otimes 1$ & 0.625\\
$3\otimes 3\otimes 1\otimes 1\otimes 1\otimes 1\otimes 1$ & 0.609375\\
$2\otimes 2\otimes 2\otimes 1\otimes 1\otimes 1\otimes 1\otimes 1$ & 0.578125\\
$3\otimes 2\otimes 1\otimes 1\otimes 1\otimes 1\otimes 1\otimes 1$ & 0.53125\\
$4\otimes 1\otimes 1\otimes 1\otimes 1\otimes 1\otimes 1\otimes 1$ & 0.5\\
$2\otimes 2\otimes 1\otimes 1\otimes 1\otimes 1\otimes 1\otimes 1\otimes 1$ & 0.4375\\
$3\otimes 1\otimes 1\otimes 1\otimes 1\otimes 1\otimes 1\otimes 1\otimes 1$ & 0.375\\
$2\otimes 1\otimes 1\otimes 1\otimes 1\otimes 1\otimes 1\otimes 1\otimes 1\otimes 1$ & 0.25\\
$1\otimes 1\otimes 1\otimes 1\otimes 1\otimes 1\otimes 1\otimes 1\otimes 1\otimes 1\otimes 1$ & 0
    \end{tabular}
\end{center}
\end{small}
\newpage

\begin{tiny}
\begin{center}
    \begin{tabular}{c|c}
        Structure & Max CE  \\
         \hline\\
         $12$ & 0.94287109375\\
$1 1\otimes 1$ & 0.923828125\\
$10\otimes 2$ & 0.923828125\\
$6\otimes 6$ & 0.9208984375\\
$7\otimes 5$ & 0.917236328125\\
$9\otimes 3$ & 0.9169921875\\
$8\otimes 4$ & 0.9140625\\
$8\otimes 2\otimes 2$ & 0.9033203125\\
$9\otimes 2\otimes 1$ & 0.900390625\\
$10\otimes 1\otimes 1$ & 0.8984375\\
$7\otimes 3\otimes 2$ & 0.89654541015625\\
$6\otimes 5\otimes 1$ & 0.89453125\\
$6\otimes 4\otimes 2$ & 0.89453125\\
$5\otimes 5\otimes 2$ & 0.89453125\\
$8\otimes 3\otimes 1$ & 0.892578125\\
$6\otimes 3\otimes 3$ & 0.89013671875\\
$7\otimes 4\otimes 1$ & 0.8896484375\\
$5\otimes 4\otimes 3$ & 0.8828125\\
$6\otimes 2\otimes 2\otimes 2$ & 0.88134765625\\
$7\otimes 2\otimes 2\otimes 1$ & 0.8758544921875\\
$4\otimes 4\otimes 4$ & 0.875\\
$8\otimes 2\otimes 1\otimes 1$ & 0.87109375\\
$6\otimes 3\otimes 2\otimes 1$ & 0.8681640625\\
$5\otimes 3\otimes 2\otimes 2$ & 0.8681640625\\
$9\otimes 1\otimes 1\otimes 1$ & 0.8671875\\
$7\otimes 3\otimes 1\otimes 1$ & 0.862060546875\\
$6\otimes 4\otimes 1\otimes 1$ & 0.859375\\
$5\otimes 5\otimes 1\otimes 1$ & 0.859375\\
$5\otimes 4\otimes 2\otimes 1$ & 0.859375\\
$4\otimes 4\otimes 2\otimes 2$ & 0.859375\\
$5\otimes 3\otimes 3\otimes 1$ & 0.853515625\\
$4\otimes 3\otimes 3\otimes 2$ & 0.853515625\\
$3\otimes 3\otimes 3\otimes 3$ & 0.847412109375\\
$4\otimes 4\otimes 3\otimes 1$ & 0.84375\\
$6\otimes 2\otimes 2\otimes 1\otimes 1$ & 0.841796875\\
$5\otimes 2\otimes 2\otimes 2\otimes 1$ & 0.841796875\\
$4\otimes 2\otimes 2\otimes 2\otimes 2$ & 0.841796875\\
$3\otimes 3\otimes 2\otimes 2\otimes 2$ & 0.835205078125\\
$7\otimes 2\otimes 1\otimes 1\otimes 1$ & 0.83447265625\\
$8\otimes 1\otimes 1\otimes 1\otimes 1$ & 0.828125\\
$6\otimes 3\otimes 1\otimes 1\otimes 1$ & 0.82421875\\
$5\otimes 3\otimes 2\otimes 1\otimes 1$ & 0.82421875\\
$4\otimes 3\otimes 2\otimes 2\otimes 1$ & 0.82421875\\
$2\otimes 2\otimes 2\otimes 2\otimes 2\otimes 2$ & 0.822021484375\\
$3\otimes 3\otimes 3\otimes 2\otimes 1$ & 0.81689453125\\
$5\otimes 4\otimes 1\otimes 1\otimes 1$ & 0.8125\\
$4\otimes 4\otimes 2\otimes 1\otimes 1$ & 0.8125\\
$4\otimes 3\otimes 3\otimes 1\otimes 1$ & 0.8046875\\
$3\otimes 2\otimes 2\otimes 2\otimes 2\otimes 1$ & 0.80224609375\\
$6\otimes 2\otimes 1\otimes 1\otimes 1\otimes 1$ & 0.7890625\\
$5\otimes 2\otimes 2\otimes 1\otimes 1\otimes 1$ & 0.7890625\\
$4\otimes 2\otimes 2\otimes 2\otimes 1\otimes 1$ & 0.7890625\\
$3\otimes 3\otimes 2\otimes 2\otimes 1\otimes 1$ & 0.7802734375\\
$7\otimes 1\otimes 1\otimes 1\otimes 1\otimes 1$ & 0.779296875\\
$5\otimes 3\otimes 1\otimes 1\otimes 1\otimes 1$ & 0.765625\\
$4\otimes 3\otimes 2\otimes 1\otimes 1\otimes 1$ & 0.765625\\
$2\otimes 2\otimes 2\otimes 2\otimes 2\otimes 1\otimes 1$ & 0.7626953125\\
$3\otimes 3\otimes 3\otimes 1\otimes 1\otimes 1$ & 0.755859375\\
$4\otimes 4\otimes 1\otimes 1\otimes 1\otimes 1$ & 0.75\\
$3\otimes 2\otimes 2\otimes 2\otimes 1\otimes 1\otimes 1$ & 0.736328125\\
$6\otimes 1\otimes 1\otimes 1\otimes 1\otimes 1\otimes 1$ & 0.71875\\
$5\otimes 2\otimes 1\otimes 1\otimes 1\otimes 1\otimes 1$ & 0.71875\\
$4\otimes 2\otimes 2\otimes 1\otimes 1\otimes 1\otimes 1$ & 0.71875\\
$3\otimes 3\otimes 2\otimes 1\otimes 1\otimes 1\otimes 1$ & 0.70703125\\
$4\otimes 3\otimes 1\otimes 1\otimes 1\otimes 1\otimes 1$ & 0.6875\\
$2\otimes 2\otimes 2\otimes 2\otimes 1\otimes 1\otimes 1\otimes 1$ & 0.68359375\\
$3\otimes 2\otimes 2\otimes 1\otimes 1\otimes 1\otimes 1\otimes 1$ & 0.6484375\\
$5\otimes 1\otimes 1\otimes 1\otimes 1\otimes 1\otimes 1\otimes 1$ & 0.625\\
$4\otimes 2\otimes 1\otimes 1\otimes 1\otimes 1\otimes 1\otimes 1$ & 0.625\\
$3\otimes 3\otimes 1\otimes 1\otimes 1\otimes 1\otimes 1\otimes 1$ & 0.609375\\
$2\otimes 2\otimes 2\otimes 1\otimes 1\otimes 1\otimes 1\otimes 1\otimes 1$ & 0.578125\\
$3\otimes 2\otimes 1\otimes 1\otimes 1\otimes 1\otimes 1\otimes 1\otimes 1$ & 0.53125\\
$4\otimes 1\otimes 1\otimes 1\otimes 1\otimes 1\otimes 1\otimes 1\otimes 1$ & 0.5\\
$2\otimes 2\otimes 1\otimes 1\otimes 1\otimes 1\otimes 1\otimes 1\otimes 1\otimes 1$ & 0.4375\\
$3\otimes 1\otimes 1\otimes 1\otimes 1\otimes 1\otimes 1\otimes 1\otimes 1\otimes 1$ & 0.375\\
$2\otimes 1\otimes 1\otimes 1\otimes 1\otimes 1\otimes 1\otimes 1\otimes 1\otimes 1\otimes 1$ & 0.25\\
$1\otimes 1\otimes 1\otimes 1\otimes 1\otimes 1\otimes 1\otimes 1\otimes 1\otimes 1\otimes 1\otimes 1$ & 0\\
    \end{tabular}
\end{center}
\end{tiny}

\section{Robustness of CE}

\label{Sect:robustness}
Suppose we apply $n$ parallelized SWAP tests on an $n$-qubit mixed state $\rho$.  The probability of getting all zero outcomes on the ancilla measurement is still given by $\Tr[\Pi_+^{\otimes n}\rho^{\otimes 2}]$.  This motivates us to define the CE of an $n$-qubit mixed state as
\begin{align}
    \mc{C}(\rho):=\mc{C}_{\ket{\psi_\rho}}([n]),
\end{align}
where $\ket{\psi_\rho}$ is a purification of $\rho$ (which is no more than $2n$ qubits), and $[n]$ is the first $n$ qubits of $\ket{\psi_\rho}$ whose reduced state is $\rho$.  We stress that this quantity is not a true entanglement monotone, as monotonicity would require us to take a convex roof extension.  However, we adopt this definition so that $\mc{C}(\rho)=1-\Tr[\Pi_+^{\otimes n}\rho^{\otimes 2}]$, and therefore $\mc{C}(\rho)$ is a directly measurable quantity using SWAP tests.

We will now use this construction to extend our concentratable hierarchies to mixed states.  For simplicity we show the idea for biseparability, but its extension to separability for more fine-grained partitions is straightforward.
\begin{proposition}
Every $n$-qubit biseparable density matrix $\rho^{A|B}$ with $|A|=k$ and $|B|=n-k$ satisfies the following inequality
\begin{equation}
    \mc{C}(\rho)\leq \mc{C}^*(k)+\mc{C}^*(n-k)- \mc{C}^*(k)\mc{C}^*(n-k)+2\sqrt{S_L(\rho)},\notag
\end{equation}
where $S_L(\rho)=1-\Tr[\rho^2]$ is the linear entropy of $\rho$.
\end{proposition}
\begin{proof}
Suppose that $\rho$ is separable with respect to some bipartition $A|B$.  Let $\ket{\alpha,\beta}$ be an optimal product state such that $\Lambda_{\max}(\rho)=\bra{\alpha,\beta}\rho\ket{\alpha,\beta}$.  Then
\begin{align}
    \mc{C}(\rho)&\leq \mc{C}(\ket{\alpha,\beta})+\Vert\rho-\op{\alpha,\beta}{\alpha,\beta}\Vert_1\notag\\
    &\leq \mc{C}(\ket{\alpha,\beta})+2\sqrt{1-\bra{\alpha,\beta}\rho\ket{\alpha,\beta}}\notag\\
   &=\mc{C}(\ket{\alpha,\beta})+2\sqrt{1-\Lambda_{\max}(\rho)},\notag
\end{align}
where we have used the Lemma  below.  Observe that
\[\Lambda_{\max}(\rho)=\max_{\sigma\in\text{SEP(A:B)}}\Tr[\sigma\rho]\geq \Tr[\rho^2]=1-S_L(\rho).\]
Substituting this to continue the previous inequality yields the desired result.
\end{proof}

\begin{lemma}
\begin{equation}
    |\mc{C}(\rho)-\mc{C}(\sigma)|\leq \Vert \rho-\sigma\Vert_1.
\end{equation}
\end{lemma}
\begin{proof}
From the definition, we have
\begin{align}
    |\mc{C}(\rho)-\mc{C}(\sigma)|&=\left|\Tr\left[\Pi_+^{\otimes n}(\rho^{\otimes 2}-\sigma^{\otimes 2})\right]\right|\notag\\
    &\leq\frac{1}{2}\left(\Vert \rho^{\otimes 2}-\sigma^{\otimes 2}\Vert_1+\Tr\left[\rho^{\otimes 2}-\sigma^{\otimes 2}\right]\right)\notag\\
    &=\Vert \rho-\sigma\Vert_1,
\end{align}
where the last line uses the facts that $0=\Tr\left[\rho^{\otimes 2}-\sigma^{\otimes 2}\right]$ and
\begin{align}
    \Vert \rho^{\otimes 2}-\sigma^{\otimes 2}\Vert_1&=\Vert \rho\otimes(\rho-\sigma)+(\rho-\sigma)\otimes\sigma\Vert_1\notag\\
    &\leq \Vert\rho-\sigma\Vert_1\left(\Vert\rho\Vert_1+\Vert\sigma\Vert_1\right)\notag\\
    &=2\Vert\rho-\sigma\Vert_1.
\end{align}
\end{proof}